\documentclass{amsart} 
\usepackage{amsmath, amsfonts, amssymb, amsthm, hyperref,verbatim}
\usepackage{xcolor}
\usepackage{braket}
\usepackage{booktabs}
\usepackage[shortlabels]{enumitem}

\newtheorem{thm}{Theorem}[section]
\newtheorem{prop}[thm]{Proposition}
\newtheorem{algorithm}[thm]{Algorithm}
\newtheorem{lem}[thm]{Lemma}
\newtheorem{cor}[thm]{Corollary}
\newtheorem{claim}{Claim}

\theoremstyle{definition}
\newtheorem{defn}[thm]{Definition}
\newtheorem{rmk}[thm]{Remark}
\newtheorem{problem}[thm]{Problem}
\newtheorem{assumption}[thm]{Assumption}
\newtheorem{exa}[thm]{Example}
\newcommand{\basis}[1]{\ket{\psi_{#1}}}
\newcommand{\cobasis}[1]{\bra{\psi_{#1}}}
\newcommand{\C}{\mathbb{C}}
\newcommand{\Q}{\mathbb{Q}}
\newcommand{\R}{\mathbb{R}}
\newcommand{\F}{\mathbb{F}}

\newcommand{\Z}{\mathbb{Z}}
\newcommand{\onto}{\twoheadrightarrow}

\renewcommand{\epsilon}{\varepsilon}
\newcommand{\eps}{\epsilon}
\newcommand{\E}{\mathop{\mathbf E}\limits}

\newcommand{\nrd}{\mathrm{nrd}}
\newcommand{\co}{\mathcal{O}}
\newcommand{\cl}{\mathrm{Cls}}

\newcommand{\ib}{\mathcal{I}}
\newcommand{\poly}{\mathrm{poly}}
\newcommand{\polylog}{\mathrm{polylog}}

\def\map#1{\xrightarrow{#1}}

\def\isom{\map{\sim}}
\newcommand{\pp}{\mathsf{PP}}
\newcommand{\ver}{\mathsf{Verify}}
\newcommand{\mint}{\mathsf{Mint}}
\newcommand{\dsk}{\mathsf{SK}}
\newcommand{\dvk}{\mathsf{VK}}
\newcommand{\ct}{\mathsf{CT}}
\newcommand{\storm}{\mathsf{Storm}}
\newcommand{\querycomp}{{\Omega}((N/\log(N))^{1/3})}

\numberwithin{equation}{section}

\begin{document}

\title[Quantum Money from Quaternion Algebras]{Quantum Money from Quaternion Algebras}
\author[D.\ M.\ Kane]{Daniel M.\ Kane}
\address{Mathematics Department, UCSD, La Jolla, CA 92093}
\email{dakane@ucsd.edu}
\author[S.\ Sharif]{Shahed Sharif}
\address{Mathematics Department, California State University, San Marcos, CA 92096}
\email{ssharif@csusm.edu}
\author[A.\ Silverberg]{Alice Silverberg}
\address{Mathematics Department, University of California, Irvine, CA 92697}
\email{asilverb@uci.edu}


\maketitle

\begin{abstract}
We propose a new idea for public key quantum money and quantum lightning. In the abstract sense, our bills are encoded as a joint eigenstate of a fixed system of commuting unitary operators. We perform some basic analysis of this black box system and show that it is resistant to black box attacks. In order to instantiate this protocol, one needs to find a cryptographically complicated system of computable, commuting, unitary operators. To fill this need, we propose using Brandt operators acting on the Brandt modules associated to certain quaternion algebras. We explain why we believe this instantiation is likely to be secure.
\end{abstract}

\section{Introduction}

One of the main challenges to building a purely digital currency is that digital information can be copied, allowing adversaries to duplicate bills or more generally perform double spending attacks. Existing cryptocurrencies solve this problem by maintaining a tamper-proof ledger of all transactions to ensure that the same bill is not spent multiple times by the same actor. Essentially, in these schemes, money is not represented by a digital token so much as a number on this decentralized ledger.

Another idea for solving the bill copying problem is to make use of the quantum no-cloning principle and taking advantage of the idea that quantum information in general \emph{cannot} be copied. A scheme to take advantage of this was proposed by Wiesner in \cite{PrivateKeyMoney}. 
In his scheme, 
the bank prepared a quantum state that was an eigenstate in a secret basis. The bank could verify the correctness of the state, but it was information-theoretically impossible for an adversary without possession of this secret to copy the state in question. Unfortunately, this scheme has the disadvantage that one needs to contact the bank in order to verify the legitimacy of a bill.

Since then, there has been an effort to develop schemes for public key quantum money---that is, a scheme by which there is a publicly known protocol for checking the validity of a bill. In such a system, the bank has a mechanism for producing valid bills, and there is a publicly known mechanism that, with high probability, non-destructively checks the validity of a given bill. It should be computationally infeasible to produce $n+1$ valid bills, given access to $n$
valid bills,
without access to the bank's secret information. Such schemes can at best be computationally secure rather than information theoretically secure, as it is a finite computational problem to construct a quantum state that reliably passes the publicly known verification procedure.

There have been several proposals over the years for cryptographically secure quantum money.
The scheme proposed by Aaronson in \cite{Aaronson} was broken by Lutomirski et al.\ in \cite{LutomirskiEtAl}. 
The scheme proposed by Farhi et al.\ in  \cite{knots} was based on knot theory. It did not have a security proof, and while it has not been broken, as pointed out by Peter Shor in \cite{ShorSimonsTalk} it is not clear that many people have tried to attack it.
The scheme proposed by Aaronson and Christiano in \cite{obfuscation} is based on hidden subspaces. While the black box model was proved secure, the security of the proposed instantiation using low degree polynomials was based on a non-standard assumption, and both that assumption and the associated scheme have been broken \cite{PenaEtAl,AaronsonBlog}.
A fix by Zhandry \cite{Lightning} showed that the Aaronson-Christiano scheme could be instantiated if one has an efficient indistinguishability obfuscator, but the quantum security of such obfuscators is unclear. Zhandry's paper \cite{Lightning} also proposed a new quantum money scheme, but the security of that scheme was called into question in a paper of Roberts \cite{Roberts}, which shows that the proof of security does not hold since the hardness assumption is false.

\subsection{Our contributions}
\label{sec:motiv-proposed}

We give a new proposal for public-key quantum money, in which a note is a tensor product $\ket{\psi}\ket{\psi}$ of a simultaneous eigenvector $\ket{\psi}$ for a finite set of commuting unitary operators $U_1,U_2,\ldots,U_t$, and its serial number is the vector of eigenvalues for $\ket{\psi}$. One can easily verify such a state is a valid bill and measure the corresponding eigenvalues of the $U_j$ non-destructively.
We show (in Corollary~\ref{cor:black-box-security} below) that this quantum money scheme is secure if the $U_j$ are implemented as oracles.

We also use these ideas to construct a scheme for \emph{quantum lightning} (see \cite{Lightning}). 
As noted in \cite{Lightning}, quantum lightning has a number of applications, including not only quantum money, but also verifiable randomness and blockchain-less cryptocurrency.

In order to implement our quantum money system securely, we need to instantiate it with an explicit set of commuting operators $U_1,...,U_t$ that are cryptographically complicated in the sense that solving Problem \ref{attackProblem} below for these specific operators is not much easier than solving it for black box operators defined on a space of the same dimension. One might expect  this to be the case for operators $U_j$ that are complicated enough that no algorithm can effectively take advantage of knowing their structure. The task of finding such operators is 
exacerbated
by the fact that in order to be commuting operators, the $U_j$ must be highly structured.

We propose using a collection of operators that comes from Brandt matrices acting on Brandt modules
associated to certain quaternion algebras. We reduce the security of the instantiated scheme to the problem of finding
a state of the form $\ket{\psi}\ket{\psi}\ket{\psi}$, where $\ket{\psi}$ is a simultaneous eigenvector of the Brandt operators.

Our paper can be viewed as an extended version of Kane's unpublished preprint \cite{Kane}.
While the title of \cite{Kane} refers to modular forms, the proposed scheme
did not use modular forms, but rather used quaternion algebras and Brandt operators,
which explains the change in title.
In his March 27, 2020 lecture at the Simons Institute for the Theory of Computing \cite{ShorSimonsTalk},
Peter Shor listed Kane's quantum money scheme in \cite{Kane} as one of very few quantum money proposals that has not yet been broken.

\subsection{Our choice of operators}
Quaternions have a long history and have been well studied by mathematicians and physicists, ever since they were discovered by Hamilton in 1843 (long before the advent of quantum computing or public key cryptography).
The Brandt matrices are well known in number theory. In particular, for a prime number $N$ and prime numbers $p$, the Brandt matrices $T(p)$ are a collection of commuting, self-adjoint matrices on a 
complex vector 
space $V_N$ of dimension approximately $N/12$, defined in \S\ref{MpSection} below. Our construction will make use of unitary operators $U = e^{iT(p)/\sqrt{p}}$. 
By the known theory of these operators, it is natural to model them 
as if they were random (see Remark~\ref{remk616}). 
This suggests that any structure that comes from these objects may be hard to exploit, and that our black box lower bounds might be indicative of their complexity.

Our instantiation could equally as well be described in the language of quaternion algebras, the language of supersingular elliptic curves, or the language of modular forms. 
Namely,
there is an equivalence of categories between isomorphism classes of supersingular elliptic curves over ${\overline{\F}_N}$ and classes of left ideals in a fixed maximal order in a quaternion algebra of discriminant $N$
(see for example \cite[Chapter 42]{Voight}).
If $p \ne N$,  
then the Brandt matrix $T(p)$ acting on the vector space $V_N$ 
is the adjacency matrix for the (directed) $p$-isogeny graph of supersingular elliptic curves over $\overline{\F}_N$.
Isogenies of elliptic curves have been well studied by mathematicians for many years. 
Such $p$-isogeny graphs of supersingular elliptic curves have recently attracted significant attention from the cryptography community, as a possible means to obtain cryptography secure against quantum attacks.

Further, the system of Brandt matrices $T(p)$ 
is isomorphic to the system of Hecke operators $T_p$ acting on the space $S_2(\Gamma_0(N))$
of weight two cusp forms of level $N$ (see \cite{Pizer,Mestre}).
Modular forms are spaces of highly symmetric analytic functions on the upper half of the complex plane with a storied mathematical history, finding applications in problems as diverse as the computation of partition numbers and the proof of Fermat's Last Theorem.

When one tries to efficiently compute with modular forms or with supersingular elliptic curves, one in fact uses quaternion algebras and Brandt matrices, which is why we phrase our instantiation in that language.

\subsection{Outline}
We 
give the black box version of our protocol in \S\ref{BlackBoxSection}, the related security problem in \S\ref{sec:blackboxsecprob}, and a proof of black box security in \S\ref{sec:blackboxsecpf}.

We construct a quantum lightning scheme in \S\ref{sec:qlightning}.

In \S\ref{sec:quatalgs} we give details of our instantiation using quaternion algebras.
We give the relevant parts of the theory of quaternion algebras 
in \S\ref{QuaternionSection}, using
\cite{Voight} as a reference.
We introduce Brandt matrices in \S\ref{MpSection},
obtain canonical encodings of ideal classes
in \S\ref{CanonicalEncodingSection} that help to make the Brandt matrices
computationally tractable,
 and give additional information about the Brandt operators in \S\ref{BrandtMatrixComputationSect}.
 We give an efficient algorithm to produce a maximally entangled state in \S\ref{bellSection}.
The protocol is formally instantiated in \S\ref{sec:inst-prot}.

In \S\ref{attacksSection} we discuss the security of the instantiation.
In \S\ref{reductionss} we reduce the security of the instantiation to the hardness of
Problem~\ref{quat-Problem}, while Sections \ref{OtherUjss} to \ref{ModularFormsSection} give possible avenues of attack and why we do not expect them to succeed.

\section{The Black Box Protocol}
\label{BlackBoxSection}

  A \emph{quantum money protocol} consists of a set $B$ of \emph{bills}, an efficient \emph{verification algorithm} $\ver$, and an efficient \emph{minting algorithm} $\mint$. The verification algorithm takes as input public parameters $\pp$ and a candidate bill $x$, and outputs True if and only if $x \in B$. The minting algorithm takes as input public and private parameters and outputs a bill $x \in B$.

Suppose $V$ is an $N$-dimensional complex vector space, and $U_1, \ldots, U_t$ are commuting unitary operators on $V$. Since the $U_j$'s commute, there exists a simultaneous eigenbasis $\{\basis{i}\}_{i=1}^N$.
We assume that we have a real eigenbasis, where ``real" means fixed by complex conjugation. 
Let $z_{ij}$ denote the eigenvalue of $U_j$ associated to the eigenvector $\basis{i}$, that is,
$U_j \basis{i} = z_{ij} \basis{i}$. 
Note that the $z_{ij}$ are complex numbers with norm $1$.
Set $v_i = (z_{i1}, \ldots, z_{it})$, the vector of eigenvalues for $\basis{i}$.
The reader should think of $N$ as being exponential and $t$ as being polynomial in a security parameter.

\begin{defn}
If $\eps \in \R^{> 0}$, we say that an eigenbasis $\{\basis{i}\}_i$ is {\bf $\eps$-separated} if $|v_k - v_j| \geq \eps$ in the $L_2$-norm whenever $j \neq k$.
\end{defn}

Given an oracle that can compute controlled versions of the $U_j$, we present the following quantum money protocol:

The public parameters consist of:
  \begin{itemize}
      \item an efficient digital signature algorithm and a verification key $\dvk$,
      \item an $N$-dimensional complex vector space $V
    $ along with a computationally feasible basis for $V$,
      \item commuting unitary operators $U_1,U_2, \ldots, U_t$ on $V$ that have a real eigenbasis, and
      \item a positive real number $\eps$.
  \end{itemize}
  Assume there is an $\eps$-separated real eigenbasis $\{\basis{i}\}_{i=1}^N$. For each $i$, let $v_i$ be the vector of eigenvalues for $\basis{i}$, as above. Then a \emph{bill} consists of a triple $(\ket{\psi},v,\sigma)$, called respectively the \emph{note}, \emph{serial number}, and \emph{signature}, given as follows:
  \begin{itemize}
      \item the note $\ket{\psi}$ is $\basis{k} \otimes \basis{k}$ for some $k$,
      \item the serial number $v$ is classical information providing an approximation of $v_k$ to error less than $\eps/3$, and
      \item $\sigma$ is
a digital signature of $v$ signed with the signing key $\dsk$ that corresponds to the verification key $\dvk$.
  \end{itemize}
  The verification algorithm $\ver(\pp, (\ket{\psi},v,\sigma))$ is as follows:
  \begin{enumerate}[(i)]
      \item Verify the digital signature $\sigma$ of $v$.
      \item For each $j=1,...,t$, use phase estimation to verify that the note $\ket{\psi}$ is an eigenstate of $U_j \otimes I_N$ and of $I_N \otimes U_j$ with eigenvalues within $\eps/2$ of those given by the entries of the serial number $v$ (where $I_N$ is the $N \times N$ identity matrix.
  \end{enumerate}
The minting algorithm $\mint(\pp, \dsk)$ is as follows:
  \begin{enumerate}[(i)]
      \item Prepare a maximally entangled state
      $\displaystyle{\frac{1}{\sqrt{N}} \sum_{i=1}^N \basis{i}\basis{i}}$ for $V$ that is the uniform superposition over all notes.
      \item   Apply phase estimation with $U_j \otimes I_N$ for each $j$. Set $\ket{\psi}$ to be the resulting state, and set the $j$th entry of the serial number $v$ to be an approximation to the eigenvalue.
      \item Set $\sigma$ to be the digital signature of $v$ with signing key $\dsk$.
  \end{enumerate}

\begin{rmk} \label{ProtocolRemarks}
There are a few important things to note about this protocol:
\begin{enumerate}[(i)]
    \item The separation assumption implies that, up to scalar multiple, the eigenbasis $\{\basis{i}\}$ is unique. Therefore the verification algorithm is correct.
    \item If the bill is valid, verification does not change it.
    \item If the note was not an eigenstate of all the $U_j \otimes I_N$ and  $I_N \otimes U_j$ before applying the verification algorithm, it will be after the phase estimation step.
    \item Due to the assumed separation of $\{\basis{i}\}$, every pair of bills that validate for the same serial number must (after verification) have notes that are the same eigenstate.
\item
If the serial number is required to be an appropriate unique rounding of the eigenvalues of $\basis{i}$ rather than merely an approximation, this looks very much like a protocol for \emph{quantum lightning} in the sense of \cite{Lightning}, that is, a mechanism that can produce and label one of a number of states but for which it is hard even for an adversarial algorithm to produce multiple copies of that state. We chose to use arbitrary approximations so that one does not need to worry about precision errors if the true eigenvalues are near the boundary between two different roundings. 
With some care, this \emph{can} be turned into a protocol for quantum lightning (see \S\ref{sec:qlightning}). 
\end{enumerate}
\end{rmk}

\section{The Security Problem}
\label{sec:blackboxsecprob}

What might an attack against this scheme look like? For quantum lightning, an attack would require a method for producing two copies of the same bolt (in this case a pair of identical eigenstates). We argue that any attack on our quantum money protocol should be able to do this. In fact it is enough to note that having four copies of the same eigenstate, one can throw away one to get three copies. Thus, we base our security on the following problem:
\begin{problem}
\label{attackProblem}
  Given $N$, a complex vector space $V \cong \C^{N}$, and commuting unitary operators $U_1, \ldots, U_t$ on $V$ that have a real eigenbasis, output a state of the form $\ket{\psi}\ket{\psi}\ket{\psi}$, where $\ket{\psi}$ is an eigenvector of all the $U_j$.
\end{problem}
The black box version of the problem is when the adversary only has black box access to the $U_j$.

The formal hardness assumption associated to Problem~\ref{attackProblem} is given below in Assumption~\ref{defn:security-assumption}. Corollary~\ref{cor:security-reduction} shows that, given a secure signature scheme, security of the quantum money protocol reduces to the Assumption.

\begin{rmk}
{Problem~\ref{attackProblem} is significantly different from the problem of producing an eigenstate of the form $\ket{\psi}\ket{\psi}$. To see this, first note that if each of $\{\ket{\psi_1}, \ldots, \ket{\psi_N}\}$ and $\{\ket{\rho_1},\ldots,\ket{\rho_N}\}$ is a real orthonormal basis for $V$, then
\begin{equation}
  \sum_{i=1}^N \ket{\psi_i} \otimes \ket{\psi_i} = \sum_{i=1}^N \ket{\rho_i} \otimes \ket{\rho_i}\label{eq:double-eigenstate-sum}
\end{equation}
by Lemma~\ref{lem:gener-double-eigenst} in Appendix~\ref{AppA}.
Suppose $\{\ket{\psi_i}\}$ is a real basis of eigenstates. Even if the basis $\{\ket{\psi_i}\}$ is unknown, we can produce a state of the form $\ket{\psi_i} \ket{\psi_i}$ by  
choosing \emph{any} real orthonormal basis $\{\ket{\rho_i}\}$ for $V$, computing the superposition $\sum \ket{\rho_i}\ket{\rho_i}$, and then measuring with respect to the $U_i \otimes U_j$. So producing a double eigenstate is easy.
There is no relation analogous to \eqref{eq:double-eigenstate-sum} for triple eigenstates $\ket{\psi_i} \ket{\psi_i}\ket{\psi_i}$, and hence no efficient method for producing a triple eigenstate ``from scratch''; for a precise statement, see Lemma~\ref{TriorthogonalClaim}. Also, simply repeating the procedure for double eigenstates and hoping to produce the same double eigenstate twice is unlikely to succeed: the double eigenstate procedure produces $\ket{\psi_i}\ket{\psi_i}$ for random $i$, and if $N$ is sufficiently large then with high probability the states produced will be distinct. See also \S{}\ref{sec:sqrtn-attack}.}
\end{rmk}

We claim that any agent capable of attacking the black box protocol must be capable of solving Problem \ref{attackProblem}. In particular, we consider three kinds of attacks on the system:
\begin{enumerate}[(i)]
\item Attacks by the mint: This would apply for systems where the mint creates a public registry of valid serial numbers (or perhaps puts them into a hash tree, publishing only the root). In such a system, the mint itself might try to cheat by creating multiple copies of bills appearing in the registry.
\item Attacks by others: An attacker given access to some number of valid bills and perhaps a much larger number of valid serial number signatures finds some procedure to spend more bills than they initially had access to.
\item Attacks on random instances: An attacker, for a random public/private key pair for the digital signature scheme and making some number of calls to a signing oracle, finds some procedure to spend more bills than they initially had access to.
\end{enumerate}

Theorem \ref{thm:forgery-attackproblem} shows that for our quantum money protocol, the three types of attacks can be reduced to solving Problem~\ref{attackProblem}. Note the similarity between its proof and the security proof in \cite[Theorem 14]{obfuscation}.

\begin{thm}\label{thm:forgery-attackproblem}
\begin{enumerate}[(i)] 
    \item
If an adversary using a quantum computer and given the secret key to the signing protocol can in time $T$  run a procedure that with probability at least $p$ produces $n+1$ valid bills with at most $n$ total serial numbers among them, then the adversary can with constant positive probability solve Problem \ref{attackProblem} in time $O(T/p)$.
  \item If an adversary, given $n$ bills and $s$ uniformly random valid signatures of serial numbers, but without access to the signing key,
can in time $T$ run a procedure that with probability at least $p$ produces $n+1$ bills that pass the verification procedure, then the adversary in time $O(T)$ with probability $p$
can, given $n+s$ calls to an oracle for the signing algorithm, either:
\begin{itemize}
    \item produce a new valid signature without access to the signing key, or
    \item solve Problem~\ref{attackProblem}.
\end{itemize}
  \item Suppose there is a quantum algorithm that for a random instantiation of the
quantum money protocol (i.e., a random choice of parameters for the
digital signature scheme, but without access to the signing key), given $n$ uniformly random bills and
the signatures corresponding to $s$ other uniformly random bills, can
generate $n+1$ valid bills
in time $T$ with probability $p$ (with the probability taken over both the space of measurement outcomes and the set of public key/private key pairs for the signature scheme). Then either:
\begin{itemize}
    \item  there is a
quantum algorithm that for a random instantiation of the digital signature scheme and given $n+s$ calls to an oracle for the signing algorithm can in time $O(T)$ and with probability at least $\frac{p}{2}$ produce a new valid signature without access to the signing key, or
    \item  there is a
quantum algorithm that  in time $O(\frac{T+c(n+s)}{p})$ and with probability at least $\frac{1}{2}$ solves Problem~\ref{attackProblem}, where $c$ is the time required to run the minting algorithm once.
\end{itemize}
\end{enumerate}
\end{thm}

\begin{proof}
The argument for (i) is easy. By the pigeonhole principle, at least two of the bills produced must have the same serial number. Given the separation between the $v_i$, this must mean that the notes in question are both of the form $\basis{i}\basis{i}$ for the same value of $i$. Using one and a half of these, the adversary (i.e., the mint) has produced a state of the form $\basis{i}\basis{i}\basis{i}$. Thus, in time $T$ the adversary can solve Problem \ref{attackProblem} with probability at least $p$. Repeating $O(1/p)$ times yields a constant probability of success.

For (ii), the adversary can use the chosen signatures and the minting algorithm to produce $n+s$ valid bills $x_1,\ldots,x_n,y_1,\ldots,y_s$; namely, to produce a bill, the adversary can produce a maximally entangled state $\sum_{i=1}^N \basis{i} \otimes \basis{i}$, measure with respect to the operators $I \otimes U_j$ for $j\in\{1,\ldots,t\}$, and then sign the tuple of eigenvalues resulting from the measurements using a single call to the signing algorithm (where $N$ and $t$ are
as in Problem~\ref{attackProblem}).  For each $k$, let $\sigma_k$ denote the signature for bill $y_k$. By hypothesis, using $x_1,\ldots,x_n,\sigma_1,\ldots,\sigma_s$, the adversary with probability at least $p$ can in time $T$ produce $n+1$ bills that pass the verification procedure. These bills along with $y_1,\ldots,y_s$ give $n+s+1$ valid bills. Thus either the adversary has produced a valid signature that is not one of the original $n+s$ signatures of valid bills (thus producing a new signature without the private key), or at least two of the $n+s+1$ bills have the same serial number, which implies that the adversary has three copies of the same eigenstate, and the adversary has solved Problem \ref{attackProblem}.

For (iii), the desired quantum algorithm first generates an instance of the quantum money protocol, i.e., generates a public key/private key pair for the  digital signature algorithm. As in (ii), using $n+s$ calls to the signature algorithm, the algorithm with probability at least $p$ either produces a new signature without using the private key, or solves Problem~\ref{attackProblem}. If the former holds with probability at least $\frac{p}{2}$, then the first conclusion holds. Now suppose that is not the case. Then the algorithm solves Problem~\ref{attackProblem} with probability at least $p/2$. A solution to Problem~\ref{attackProblem} is independent of the signature keys. Repeat $O(1/p)$ times with $O(1/p)$ random instances of the signature key pair to obtain the second conclusion.
\end{proof}

In Corollary~\ref{cor:security-reduction}, we show that our quantum money protocol is secure under the following assumption.
\begin{assumption}\label{defn:security-assumption}
  For all quantum circuits $C$ that are of polynomial size in $\log N$, the probability that $C$ solves Problem~\ref{attackProblem}
  is negligible as a function of $\log N$.
\end{assumption}
An equivalent formulation is that the expectation
\[
  \E_{U_1,\ldots,U_t}\left[\sum_{\basis{i}} |\bra{\psi_i}\bra{\psi_i}\bra{\psi_i}C(V, U_1,\ldots, U_t)\ket{0}|^2\right]
\]
is negligible as a function of $\log N$, where the expectation is over all commuting unitary operators $U_1, \ldots, U_t$ on $V \cong \C^{N}$ and the sum is over any fixed eigenbasis $\{\basis{i}\}$ for the $U_j$.

The following definition, for which we follow \cite[Definition 9]{obfuscation}, is the usual notion of security for quantum money protocols.
\begin{defn}\label{defn:security}
  Given a quantum money protocol,
let $\ct$ be an algorithm that on input the public parameters $\pp$ and a list of $m = \polylog(N)$ possibly entangled alleged bills $y_1,\ldots,y_m$, outputs the number of these that pass the validation algorithm. 
 The quantum money scheme is \emph{secure} if the associated digital signature scheme is secure against existential forgeries, and for all quantum circuits $C$ of size $\polylog(N)$ that take as input valid bills $x_1,\ldots,x_n$ and signatures $\sigma_1,\ldots,\sigma_s$, where $n$ and $s$ are $\polylog(N)$, and outputs a set of $m = \polylog(N)$ alleged bills, the probability  
    \[
    \Pr[\ct(\pp, C(\pp,x_1,\ldots,x_n,\sigma_1,\ldots,\sigma_s)) > n]
  \]
  is a negligible function of $\log N$, where the probability is taken over all sets of public parameters,  
  valid bills, and signed serial numbers. \end{defn}

\begin{cor}\label{cor:security-reduction}
  If the digital signature scheme used in the quantum money protocol is secure against existential forgeries, and Assumption~\ref{defn:security-assumption} holds, then the quantum money protocol is secure.
\end{cor}

\begin{proof}
 Let $\ct$ be an algorithm that on input $\pp$ and a list of $m = \polylog(N)$ possibly entangled alleged bills $y_1,\ldots,y_m$, outputs the number of these that pass the validation algorithm. Suppose there is a polynomial time quantum adversary $C$ for which, for infinitely many values of $N$,
    \[
    \Pr[\ct(\pp, C(\pp,x_1,\ldots,x_n,\sigma_1,\ldots,\sigma_s)) > n] \ge 1/f(\log N)
  \]
  for some positive polynomial $f$, where the $x_i$ are valid bills, the $\sigma_j$ are valid signatures, and $n$ and $s$ are $\polylog(N)$. By (iii) of Theorem~\ref{thm:forgery-attackproblem}, either there is a polynomial time algorithm that, for infintely many $N$ produces a new signature without the signing key with probability at least $1/(2f(\log N))$ (violating the security of the signature scheme) or there is a polynomial time algorithm that for infinitely many $N$ solves Problem~\ref{attackProblem} with probability at least $\frac{1}{2}$ (violating Assumption~\ref{defn:security-assumption}).
\end{proof}

\subsection{A $\sqrt{N}$ attack}
\label{sec:sqrtn-attack}

There is an obvious $O(\sqrt{N})$ time attack
on Problem~\ref{attackProblem},
namely:
\begin{itemize}
\item Produce $\sqrt{N}$ notes using the minting procedure, and
\item Search for pairs of notes with serial numbers sufficiently close to each other.
\end{itemize}
Each note is $\basis{i}\basis{i}$ for a uniform random value of $i$. By the birthday paradox, we expect to find a collision within the first $O(\sqrt{N})$ notes.

\section{Black Box Security}
\label{sec:blackboxsecpf}

One might worry about black box attacks against the proposed system, that is, attacks on
Problem~\ref{attackProblem}
that do not make use of any special structure of $V$ or the $U_j$ and only have black box access to the operators $U_j$. In this section we will prove
Theorem~\ref{LowerBoundTheorem}, which says that any such attack must have query complexity at least $\Omega((N/\log(N))^{1/3})$. As a consequence, we show in Corollary~\ref{cor:black-box-security} that Assumption~\ref{defn:security-assumption} holds in the black box setting.

Let $S^1$ denote the unit circle $S^1 := \{ z\in\C : |z|=1 \}$.
If $\mathcal{D}$ is a probability distribution over $(S^1)^t$, then for each $N$, we obtain an induced probability distribution over tuples of commuting $N \times N$ unitary operators $(U_1,\ldots,U_t)$ by letting $\{\basis{i}\}$ be a random real orthonormal basis of $\C^N$ (under the Haar measure), letting $v_i = (z_{i1},\ldots,z_{it})$ be i.i.d.\ samples from $\mathcal{D}$ for $i=1,\ldots,N$, and defining $U_j$ by the equations $U_j\basis{i} = z_{ij}\basis{i}$.
\begin{thm}\label{LowerBoundTheorem}
    Suppose $\mathcal{D}$ is any probability distribution over $(S^1)^t$ such that with probability 1, any finite number of samples chosen from $\mathcal{D}$ are distinct.  Suppose $C$ is a circuit consisting of standard gates and $d$ controlled $U_j$ gates. If $d^3 < \frac{N}{16 \log N}$ for all sufficiently large $N$, then $C$ solves Problem~\ref{attackProblem} with probability $O(\frac{d^3\log d}{N})$, where the probability is taken over sets of operators $U_1,\ldots,U_t$ chosen according to $\mathcal{D}$ and with uniformly random real orthogonal eigenbasis.
\end{thm}
We note that if  $\mathcal{D}$ is the uniform distribution over $(S^1)^t$, then with high probability the eigenbasis will be $\eps$-separated, and thus by Theorem~\ref{thm:forgery-attackproblem} these instances will provide quantum money secure against black box attacks. However, we note that Theorem~\ref{LowerBoundTheorem} applies in greater generality, so long as the eigenspaces of the $U_i$ are non-degenerate (if they are degenerate, the problem may be easier as there will be many more vectors $\ket{\psi}$ that constitute eigenvectors). In fact we will show that a slight variation of Problem~\ref{attackProblem} (see Problem~\ref{attackProblemRefined} below) that is equivalent for non-degenerate eigenspaces is difficult unconditionally in the black box model.

 \begin{cor}\label{cor:black-box-security}
   Assumption~\ref{defn:security-assumption} holds and the quantum money protocol is secure, if the adversary has only black box access to the $U_j$.
 \end{cor}

 \begin{proof}
   Suppose a circuit $C$ with only black box access to the $U_j$ has size $d = \polylog(N)$. Then  
   for $N$ sufficiently large, $d^3 < \frac{N}{16 \log N}$. 
   By Theorem~\ref{LowerBoundTheorem}, $C$ solves Problem~\ref{attackProblem} with probability $< k\frac{d^3\log d}{N}$ for some positive constant $k$. Since $d = \polylog(N)$, $k \frac{d^3\log d}{N}$ is negligible as a function of $\log N$, so Assumption~\ref{defn:security-assumption} holds. The security of the quantum money protocol follows from Corollary~\ref{cor:security-reduction}.
 \end{proof}

The proof of Theorem \ref{LowerBoundTheorem} will proceed in three steps:
\begin{enumerate}[(i)]
    \item Replace Problem~\ref{attackProblem} with a refinement, Problem~\ref{attackProblemRefined}, that is equivalent when the eigenspaces are one-dimensional.
\item
Show that with degenerate eigenspaces  (i.e., eigenspaces of dimension greater than one), Problem~\ref{attackProblemRefined} is impossible to solve with constant positive probability even with an unbounded number of queries (with probability of success depending on how degenerate the eigenspaces are).
\item
Then define a family of input distributions parameterized by an integer $M$ so that when $M$ is large we have $\epsilon$-separation with high probability and when $M$ is small we do not. We use the bounds from (ii) to show that the probability of success with small $M$ is bounded and then use the polynomial method to show that unless we make a large number of queries, this implies that the probability of success is small even in the range where we do have $\epsilon$-separation.
\end{enumerate}

\subsection{Preliminary lemmas}
\label{sec:preliminary-lemmas}

Consider the following refinement of Problem 3.1:
\begin{problem}
\label{attackProblemRefined}
  Given $N$, a complex vector space $V \cong \C^{N}$, and commuting unitary operators $U_1, \ldots, U_t$ on $V$, output a state of the form $\basis{i}\basis{i}\basis{i}$ for some $1\leq i \leq N$, where  $\{\basis{i}\}$ is a fixed secret real eigenbasis of $V$ for the operators $U_j$.
\end{problem}
When the eigenbasis is $\epsilon$-separated, then Problems \ref{attackProblemRefined} and \ref{attackProblem} are equivalent. But if the eigenspaces are
degenerate, then Problem~\ref{attackProblemRefined} is impossible to solve. To see this, suppose that a circuit attempting to solve Problem~\ref{attackProblemRefined} outputs a state $\ket{\phi}$, and think of the choice of basis $\{\basis{i}\}$ as a random variable. Then it suffices to show that the probability that $\ket{\phi}$ has a large component in any $\basis{i}\basis{i}\basis{i}$ direction is small. We first consider the case of a single, totally degenerate eigenspace; we will consider the general case in Claim~\ref{qMilem}
of the proof of Theorem~\ref{LowerBoundTheorem} in \S\ref{Thm41pf} below.
\begin{lem}\label{TriorthogonalClaim}
If $W$ is a complex vector space and $\ket{\phi}\in W\otimes W \otimes W$, then
$$
\E_{\{\basis{i}\}\textrm{ real orthonormal basis of }W}\Biggl[\sum_{i} |\cobasis{i}\cobasis{i}\cobasis{i}\ket{\phi} |^2 \Biggr] \leq \frac{3}{\dim(W)}.
$$
\end{lem}

\begin{proof}
Let $m=\dim(W)$.
It suffices to show that
$$
\E_{|\basis{}|_2=1}\left[ |\cobasis{}\cobasis{}\cobasis{}\ket{\phi} |^2 \right] \leq \frac{3}{m^2}.
$$
Rewrite $\basis{}$ as $\frac{1}{\sqrt{m}}\sum_{i=1}^m x_i\basis{i}$ where $\{\basis{i}\}$ is a random real orthonormal basis for $W$ and $x_i$ are i.i.d.\ $\pm 1$ random variables. We claim that even after fixing the $\basis{i}$, the expectation over $x_i$ is at most $3/m^2$. In particular let $$\ket{\phi} = \sum_{1\leq i, j, k \leq m} a_{ijk} \basis{i}\basis{j}\basis{k}$$ where $\sum_{1\leq i, j, k \leq m} |a_{ijk}|^2=1$. Then the expectation over $x_i$ is
$$
\frac{1}{m^3}\E_{x_i}[| \sum_{1\leq i , j , k \leq m} a_{ijk}x_ix_jx_k|^2].
$$
Collecting like terms this is
\begin{multline*}
\frac{1}{m^3}\E_{x_i}[| \sum_{1 \leq i < j < k \leq m} (a_{ijk}+a_{ikj}+a_{jik} +a_{jki}+a_{kij}+a_{kji})x_ix_jx_k  \\
 + \sum_{i=1}^m x_i (a_{iii} + \sum_{j=1,j\neq i}^m (a_{ijj}+a_{jij}+a_{jji}))|^2].
\end{multline*}
By orthogonality of the variables $x_ix_jx_k$ and $x_i$, this is
\begin{multline} 
\label{eq:aijk}
\frac{1}{m^3}\E_{x_i}[\sum_{1 \leq i < j < k \leq m} |a_{ijk}+a_{ikj}+a_{jik}  +a_{jki}+a_{kij}+a_{kji}|^2 \\
+ \sum_{i=1}^m |a_{iii} + \sum_{j=1,j\neq i}^m (a_{ijj}+a_{jij}+a_{jji})|^2].
\end{multline} 
For each $i$, there are $3m-2$ terms in the sum
  $a_{iii} + \sum_{j=1,j\neq i}^m (a_{ijj}+a_{jij}+a_{jji}).$
Thus by Cauchy-Schwartz, \eqref{eq:aijk} is at most
$$
\frac{1}{m^3}\Biggl(\sum_{\substack{1\leq i, j, k \leq m \\ i\neq j \neq k}} 6|a_{ijk}|^2 + (3m-2)\sum_{i=1}^m \Bigl(|a_{iii}|^2 + \sum_{j=1,j\neq i}^m (|a_{ijj}|^2+|a_{jij}|^2+|a_{jji}|^2)\Bigr)\Biggr).
$$
Collecting terms, this is at most
$
\frac{1}{m^3}\sum_{i,j,k=1}^m (3m-2)|a_{ijk}|^2  \leq \frac{3}{m^2},
$
as desired.
\end{proof}

Our proof of Theorem~\ref{LowerBoundTheorem} in \S \ref{Thm41pf} will
 make use of the following two lemmas.

\begin{lem}\label{lemma:O-MoverN}
 Suppose $N,M \in \Z^{> 0}$, and $M \leq \frac{N}{16 \log N}$.
   With the probability taken over the space of all functions $h: [N] \to [M]$,
we have
  \[
    \Pr\left(\#(h^{-1}(j)) > \frac{N}{2M} \, \text{ for all $j$}\right) \geq 1 - \frac{1}{16 N \log N}.
  \]
\end{lem}

\begin{proof}
  Fix $j \in [M]$. For $i \in [N]$, define a random variable $X_i$ by
  \[
    X_i =
    \begin{cases}
      1 & h(i) = j, \\
      0 & h(i) \neq j.
    \end{cases}
  \]
  The probability that $X_i = 1$ is $\frac{1}{M}$, and the $X_i$ are independent. Let $X = \sum_{i=1}^N X_i$. Observe that $\E_h[X] = \frac{N}{M}$. By the Chernoff bounds,
  \[
    \Pr\left(X \leq \frac{N}{2M}\right) \leq e^{-\frac{N}{8M}} \leq \frac{1}{N^2},
  \]
where the last inequality holds since $M \leq \frac{N}{16 \log N}$.
By the union bound, we have
  \[
    \Pr\Bigl(\#(h^{-1}(j)) \leq \frac{N}{2M} \, \textrm{ for some } j\Bigr)
    \leq \sum_{j=1}^M \Pr\Bigl(X \leq \frac{N}{2M}\Bigr)
    \leq \frac{M}{N^2}
    \leq \frac{1}{16 N \log N}.
  \]
  The claim now follows.
\end{proof}

\begin{lem}\label{lemma:infinite-product}
  If $i \in \Z^{\geq 1}$, then
    $\displaystyle{\Biggl|\prod_{\substack{j \geq 1 \\j\neq i}} \frac{(2j-1)^2}{(2j-1)^2-(2i-1)^2}\Biggr| = O\left(\frac{1}{i}\right).}$
\end{lem}
\begin{proof}
Let $f(z) :=\prod_{j=1}^\infty \left(1-\frac{z}{(2j-1)^2}\right) = \cos(\frac{\pi \sqrt{z}}{2})$.
Then
\begin{align*}
    \Bigl|\prod_{\substack{j \geq 1 \\j\neq i}} \frac{(2j-1)^2}{(2j-1)^2-(2i-1)^2}\Bigr|
    &= \Bigl|\frac{1}{(2i-1)^{2}f'((2i-1)^2)}\Bigr| \\
                      &= O\Biggl(\frac{1}{|(2i-1)\sin \Bigl(\frac{(2i-1)\pi}{2}\Bigr)|}\Biggr) \\
    &= O\left(\frac{1}{i}\right).
\end{align*}
\end{proof}

\subsection{Proof of Theorem~\ref{LowerBoundTheorem}}
\label{Thm41pf}
Next, we use a quantum modification of the polynomial method.  
Let $C$ be any circuit consisting of standard gates and at most $d$ controlled $U_j$ gates. We show that under the correct distributions over $U_j$, any circuit with $d$ too small will be unable to distinguish the cases where the eigenspaces of $U_j$ are degenerate, and those where it is not.

Let $v_1,\ldots,v_N$ be i.i.d.\ samples from $\mathcal{D}$,
let $\{\basis{i}\}$ be a random real orthonormal basis, and
let  $(U_1,\ldots,U_t)$ be the operators determined by these choices. Since, by hypothesis, $N$ samples chosen from $\mathcal{D}$ are with high probability distinct, every solution of Problem~\ref{attackProblem} is also a solution of Problem~\ref{attackProblemRefined}, and so the probability that circuit $C$ solves Problem~\ref{attackProblem} is
$$
\E_{\basis{i},v_i}\Bigl[\sum_i |\cobasis{i}\cobasis{i}\cobasis{i}C(U_1,\ldots,U_t)\ket{0}|^2\Bigr],
$$
where the expectation is over all choices of real orthonormal basis $\{\basis{i}\}$ and tuples of eigenvalues $v_1,\ldots,v_N$.
On fixing the eigenbasis $\basis{i}$, each controlled $U_j$ gate becomes an operator with entries that are linear functions in the $z_{ij}$. Thus the entries of $C\ket{0}$ are degree $d$ polynomials in the $z_{ij}$, and 
$\sum_i |\cobasis{i} \cobasis{i} \cobasis{i}C \ket{0}|^2$ is a degree $2d$ polynomial in $z_{ij}$ and $\overline{z_{ij}}$. Taking an expectation over the eigenbasis shows that the above expectation
is of the form
$
\E_{v_i}[p(z_{ij},\overline{z_{ij}})],
$
where $p$ is some polynomial of degree at most $2d$ and
$v_i = (z_{i1}, \ldots, z_{it})$.

For integers $M$, we define a slightly different probability distribution over the $v_i$. We let $h:[N]\rightarrow [M]$ be a function chosen uniformly at random,
and let $v_i = u_{h(i)}$ where the $u_j$ are i.i.d.\ elements of $\mathcal{D}$. We let
$$
A_M = \E_{h, v_i}[p(z_{ij},\overline{z_{ij}})]
$$
where the $h$ vary uniformly among functions $[N] \to [M]$ and the $v_i$ are distributed as above, with $v_i = u_{h(i)}$ and the $u_j$ distributed according to $\mathcal{D}$.

There are several things worth noting about this distribution. First, it is easy to see that our original probability of success is $\lim_{M\rightarrow \infty}A_M$. This is because for large $M$, with high probability $h$ has no collisions and therefore the distribution over the $v_i$ is arbitrarily close in total variational distance to i.i.d.\ copies of $\mathcal{D}$.
Second, we have the following:

\begin{claim}
\label{Nqlem}
For each $N$, there exists a polynomial $q_N(x)$ of degree at most $2d$ such that $A_M = q_N(1/M)$.
\end{claim}

\begin{proof}
  Since $p(z_{ij},\overline{z_{ij}})$ is a polynomial of degree at most $2d$, to prove the claim it suffices to show that if $m$ is a monic monomial of degree $e$,
  then
  $\E_{h,u_i}[m(z_{ij},\overline{z_{ij}})]$
  is a polynomial in $1/M$ of degree at most $e$.
  Write $m(z_{ij},\overline{z_{ij}}) = m_1(z_{ij}) m_2(\overline{z_{ij}})$ with $m_1$ and $m_2$ monic. Observe that
  \[
    \E_{v_i}[m(z_{ij},\overline{z_{ij}})] =
    \begin{cases}
      1 & \text{if $m_1 = m_2$} \\
      0 & \text{if $m_1 \neq m_2$}.
    \end{cases}
  \]
  Write $u_{ij}$ for the $j$th component of $u_i$. Given $h$, define a ring homomorphism
  $
  H: \C[\{z_{ij}\}] \to \C[\{u_{ij}\}]
  $
  by $H(z_{ij}) = u_{h(i)j}$. Then
  $\E_{h,u_i}[m(z_{ij},\overline{z_{ij}})]$
  equals the probability over the set of $h$'s that
  $H \circ m_1 = H \circ m_2$. If $m_1 = m_2$, then this probability is $1$. Now suppose that $m_1 \neq m_2$. For $k = 1$ and $2$, let
  \[
    B_k = \{z_{ij} 
    : z_{ij} \text{ appears in $m_k$ with positive exponent}\}.
  \]
  Since $z_{ij}\overline{z_{ij}} = 1$ whenever $z_{ij} \in S^1$, by cancelling such
  terms
  in $m$ we may assume that $B_1 \cap B_2 = \emptyset$. Without loss of generality,
  $|B_1| \geq |B_2|$. If $t$ is a surjective map $B_1 \onto B_2$, say that $h$ has \emph{collision type $t$} if $H(z) = H(t(z))$ for all $z \in B_1$.
  There is a finite set $T$ of collision types with the property that $h$ has collision type in $T$ if and only if $H \circ m_1 = H \circ m_2$. The
  number of $h$ having
  collision type in the set $T$ is given by an inclusion-exclusion formula. Each term in the inclusion-exclusion is given by the
  number of $h$ having collision type $t$ for all $t$ in some subset $T' \subseteq T$. For a given collision type $t \in T$, the probability that $h$ has type $t$ is
  $\frac{1}{M^{|B_1|}}$. The probability that $h$ has type $t$ for every $t \in T'$ is of the form $K/M^f$ for some constant $K$ and some integer $f$. The maximum value of $f$ occurs
  for the sets $T'$ such that $h$ has type $t$ for all $t \in T'$ if and only if $H|_{B_1 \cup B_2}$ is a constant, in which case $f = |B_1| + |B_2| - 1$. Since
  $|B_1| + |B_2| \leq e$,
Claim~\ref{Nqlem}
  follows.
\end{proof}

\begin{claim}
\label{qMilem}
\medskip
If $M \le \frac{N}{16\log N}$, then $q_N(1/M) = O(M/N)$.
\end{claim}

 \begin{proof}
   By the above discussion and Claim~\ref{Nqlem}
   we have:
   \begin{equation}\label{eq:q_n-expectation}
     q_N(1/M) = \E_{h,\basis{i},v_i}\Bigl[\sum_i |\cobasis{i}\cobasis{i}\cobasis{i}C(U_1,\ldots,U_t)\ket{0}|^2\Bigr],
   \end{equation}
   where the $h$ vary uniformly over functions from $[N]$ to $[M]$, the $v_i$ are distributed according to $h$ and $\mathcal{D}$ as above, and the sets $\{\basis{i}\}$ vary over random real orthonormal bases for $V$.
   Suppose $M \le \frac{N}{16\log N}$, and let $h$ be a random function from $[N]$ to $[M]$. By  Lemma~\ref{lemma:O-MoverN},
   with probability at least $1 - \frac{1}{16N \log N}$, for every $j\in [M]$ we have $\#(h^{-1}(j)) = \Omega(N/M)$. Let $V_j = \mathrm{span}\{\basis{i}:h(i)=j\}$, so that with probability at least $1 - \frac{1}{16N \log N}$ we have $\dim V_j = \Omega(N/M)$.
   Fix both the values of the $u_j$ and the spaces $V_j$.
   The $V_j$ are eigenspaces for $U_k$ with eigenvalues $u_{jk}$. The output of $C$ depends only on the $V_j$ and the $u_j$, but not on \emph{which} basis of $V_j$ is given by $\{\basis{i}:h(i)=j\}$. Thus the output is  $\sum_j a_j \ket{\phi_j}$ for some $\ket{\phi_j}\in V_j^{\otimes 3}$ and $\sum_j |a_j|^2=1$. Therefore the right-hand side of~\eqref{eq:q_n-expectation} is
   $$
   \E_{V_j, u_j} \sum_j |a_j|^2 \Bigl[\E_{\basis{i} \text{ given } V_j}\Bigl[\sum_{i:\basis{i} \in V_j} |\cobasis{i}\cobasis{i}\cobasis{i}\ket{\phi_j} |^2 \Bigr]\Bigr].
   $$
   Here, we vary over orthogonal decompositions $V = \oplus_{j=1}^M V_j$, tuples of eigenvalues $u_j$, and real orthonormal bases $\{\basis{i}\}$ that are a union of real orthonormal bases for the $V_j$.
   Note that $h$ can be recovered from this data by defining $h(i) = j$ if and only if $\basis{i} \in V_j$. Thus varying over tuples $(\{\basis{i}\}, h, u_j)$ is the same as varying over tuples $(\{V_j\}, \{\basis{i}\} \text{ given } V_j, u_j)$, with an appropriate choice of distribution on the latter tuples.
   By Lemma~\ref{TriorthogonalClaim}, with probability at least $1 - \frac{1}{16N \log N}$ we have:
   \begin{align*}
     \sum_j |a_j|^2 \E_{\basis{i} \text{ given } V_j}\Bigl[\sum_{i:\basis{i} \in V_j} |\cobasis{i}\cobasis{i}\cobasis{i}\ket{\phi_j} |^2 \Bigr] 
     &= \sum_j |a_j|^2 O(1/\dim(V_j)) \\
     &= O(M/N).
   \end{align*}
   By \eqref{eq:q_n-expectation} we have $q_N(1/M) = O(M/N)$, as desired.
 \end{proof}

We now proceed to prove Theorem~\ref{LowerBoundTheorem}. By Claim~\ref{Nqlem},
for each $N$ the probability that $C$ solves Problem~\ref{attackProblem} is
  $\lim_{M \to \infty} A_M = q_N(0)$,
  so it suffices to show that $q_N(0) = O(d^3\log(d)/N)$. Let $C$ be as in the hypothesis, and take $N$ large enough that $d^3 < \frac{N}{16 \log N}$.

  For $i\in\{1,\ldots,2d+1\}$, let
$$m_i = \frac{d^3}{(2i-1)^2} \text{ and } M_i = \left\lfloor m_i \right\rfloor \in\Z.$$
Then
$M_i \le d^3 < \frac{N}{16\log N}$,
so  
$q_N(1/M_i) = O(M_i/N)$ for each $i$ by Claim~\ref{qMilem}.

Using standard polynomial interpolation,
we have:
$$
q_N(0) = \sum_{i=1}^{2d+1} q_N(1/M_i) \prod_{j\neq i} \frac{1/M_j}{1/M_j-1/M_i}.
$$
We begin by bounding these expressions if the $M_j$ were replaced by $m_j$:
\begin{multline*}
\Bigl|\prod_{j\neq i,j\leq 2d+1} \frac{1/m_j}{1/m_j-1/m_i}\Bigr|
= \Bigl|\prod_{j\neq i,j\leq 2d+1} \frac{(2j-1)^2/d^3}{(2j-1)^2/d^3-(2i-1)^2/d^3}\Bigr|\\
 = \Bigl|\prod_{j\neq i,j\leq 2d+1} \frac{(2j-1)^2}{(2j-1)^2-(2i-1)^2}\Bigr|
 \leq \Bigl|\prod_{j\neq i} \frac{(2j-1)^2}{(2j-1)^2-(2i-1)^2}\Bigr|,
\end{multline*}
where the final product is over all positive integers $j$. By Lemma~\ref{lemma:infinite-product}, the latter product is $O(1/i)$.
Since $M_i = \lfloor m_i \rfloor$, we have
$
1/M_i = 1/m_i + O(1/m_i^2).
$
Thus,
\begin{align*}
\Biggl|  \prod_{j\neq i} \frac{1/M_j}{1/M_j-1/M_i} \Biggr| & = \Biggl|  \prod_{j\neq i} \frac{1/m_j+O(1/m_j^2)}{1/m_j-1/m_i + O(1/m_i^2+1/m_j^2)} \Biggr|\\
                                                  & \leq \Biggl|  \prod_{j\neq i} \frac{1/m_j}{1/m_j-1/m_i} \Biggr|\prod_{j\neq i}\Biggl(1+ \frac{O(1/m_i^2+1/m_j^2)}{|1/m_i-1/m_j|}\Biggr)\\
& = O(1/i)\exp\Biggl(\sum_{j\neq i} O\Biggl(\frac{i^4+j^4}{(i^2-j^2)d^3} \Biggr) \Biggr)\\
& \leq O(1/i)\exp\Biggl(\sum_{j\neq i} O\Biggl(\frac{\max(i,j)^4}{(\max(i,j)|i-j|d^3} \Biggr) \Biggr)\\
& \leq O(1/i)\exp\Biggl(\sum_{j\neq i} O\Biggl(\frac{\max(i,j)^3}{|i-j|d^3} \Biggr) \Biggr).
\end{align*}
Now if $i\leq \sqrt{d}$, the terms with $j\leq 2i$ sum to at most $O(1/d)$, and the larger terms in the sum are $O\left(\frac{j^2}{d^3}\right)$, and therefore sum to $O(1)$. If $i\geq \sqrt{d}$, then the terms are $O\left(\frac{1}{|i-j|}\right)$, and thus sum to $O({\log(d)})$.
 This implies that
\begin{align*}
q_N(0) & = \sum_{i=1}^{2d+1} q_N(1/M_i) \prod_{j\neq i} \frac{1/M_j}{1/M_j-1/M_i}\\
& = \sum_{i=1}^{\sqrt{d}} q_N(1/M_i)O(1/i) + \sum_{i=\sqrt{d}}^{2d+1} q_N(1/M_i)O(\log(d)/i) \\\label{eq:q-of-zero}
 & = \sum_{i=1}^{\sqrt{d}} O\left(\frac{M_i}{N i}\right) + \sum_{i=\sqrt{d}}^{2d+1} O\left(\frac{\log(d)M_i}{Ni}\right) \\
     & = \sum_{i=1}^{\sqrt{d}} O\left(\frac{d^3}{N i^3}\right) + \sum_{i=\sqrt{d}}^{2d+1} O\left(\frac{d^3\log(d)}{Ni^3}\right) \\
     & = O(d^3\log(d)/N),
\end{align*}
as desired. \hfill $\qed$

\begin{rmk}
The bound in Theorem~\ref{LowerBoundTheorem} is nearly tight. In particular, if we assume $\epsilon$-separation of the $v_i$'s for the operators $U_1,\ldots,U_t$, then there is actually an algorithm for solving Problem \ref{attackProblem} with constant probability in $O(N^{1/3}t/\epsilon)$ queries, similar to the collision algorithm of~\cite{bht}. The algorithm involves computing $N^{1/3}$ pairs $\basis{i}\basis{i}$, then preparing $N^{2/3}$
other maximally entangled states. These maximally entangled states can be thought of as being in a superposition of all combinations of $N^{2/3}$ pairs tensored together. There is a reasonable probability that one of these $N^{2/3}$ pairs agrees with one of our $N^{1/3}$ pairs, and we can find the index of such a pair using Grover's algorithm by measuring the eigenvalues of only $O(N^{1/3})$ of our pairs. In order to compute the eigenvalues to sufficient accuracy takes only $O(t/\epsilon)$ queries each. Thus, this algorithm has query complexity $O(N^{1/3})$, although the full complexity is $O(N^{2/3})$.
\end{rmk}

\section{Quantum Lightning}
\label{sec:qlightning}

Following \cite{Lightning}, a quantum lightning protocol consists of:
\begin{itemize}
    \item a \emph{storm} $\storm$, which is a polynomial size quantum algorithm that on input a security parameter, outputs a quantum state $\ket{\psi}$ called a \emph{bolt}, and
    \item a quantum verification algorithm $\ver$ that on input a bolt,
    outputs a serial number if the bolt is ``valid'' (that is, is an output of $\storm$), and outputs $\perp$
     if the bolt is not valid,
\end{itemize}
satisfying:
\begin{itemize}
    \item the expected value of
    $-\log_2 \min_s \Pr[\ver(\ket{\psi}) = s]$
  is negligible as a function of the security parameter, where $s$ is a serial number, and where the expectation is taken over all pairs $(\storm, \ver)$ and valid bolts $\ket{\psi}$, and 
    \item the expected value of
    $1 - |\braket{\psi' | \psi}|^2$
  is a negligible function of the security parameter, where the expectation is taken over all pairs $(\storm, \ver)$ and valid bolts $\ket{\psi}$, and $\ket{\psi'}$ is the state obtained by running $\ver$ on $\ket{\psi}$.
\end{itemize}
As in \cite[Definition 3.2]{Lightning},
a quantum lightning scheme is \emph{secure} if no polynomial time adversary can, with
non-negligible probability, produce two 
states such that a verifier will read them as valid bolts with the same serial number.

It is shown in \cite{Lightning} that a protocol for quantum lightning can be turned into a quantum money protocol with comparable security, along with some other applications.

We note that the joint eigenstates of our operators $U_j$ have many of the properties of quantum lightning. While there is a ``storm'' that can produce pairs of eigenstates $\ket{\psi}\ket{\psi}$, assuming the difficulty of Problem \ref{attackProblem} it is computationally difficult to produce two copies of this state. Furthermore, the state $\ket{\psi}\ket{\psi}$ can be associated with a serial number given by the vector of its eigenvalues with respect to the $U_j$. Unfortunately, this does not quite match up with the definition in  \cite{Lightning} of quantum lightning, as these ``serial numbers'' can only be computed approximately. To fix this, we round the eigenvalues to the nearest multiples of some small $\eps$. Unfortunately, this creates issues if the true eigenvalue is very close to halfway between two such multiples. To fix this, we modify the storm to throw away bolts that are too close to this boundary. To ensure that some eigenstates are not rejected by this, we also randomize the boundary somewhat.

Our quantum lightning protocol is as follows:

\begin{enumerate}[(i)]
\item Have a set of commuting operators $U_1,U_2,\ldots,U_t$ and some $\eps>0$ so that for any distinct joint eigenstates $\ket{\psi}$ and $\ket{\rho}$ there is some $i$ so that the eigenvalues of $U_j$ on $\ket{\psi}$ and $\ket{\rho}$ differ by at least $10\eps$.
\item
Pick any $\delta$ satisfying $\eps/(10 t) > \delta > 0$, and choose a complex number $z$ uniformly at random from the unit square.
\item
Our storm generates a pair $\ket{\psi}\ket{\psi}$ of joint eigenstates as in the minting algorithm $\mint$ and computes the eigenvalues $\{\lambda_i\}$ of $\{U_j\}$
on $\ket{\psi}$. If the real or imaginary parts of any $\lambda_i + z$ are within $\delta$ of a multiple of $\eps/2$, the algorithm tries again.
\item
Our verifier takes a state $\ket{x}\ket{y}$ and measures the eigenvalues of each $U_j$ on $\ket{x}$ and $\ket{y}$  to error $\delta$, giving $\{\lambda_i\}$ and $\{\mu_i\}$. If $|\lambda_i - \mu_i| > 2\delta$ for any $i$, it rejects the input. Otherwise, it rounds the real and imaginary parts of $\lambda_i+z$ to the nearest multiple of $\eps$ and returns the vector of these values as the serial number.
\end{enumerate}

\begin{thm}
\label{thm51}
If we have an instantiation of the quantum money protocol of \S{}\ref{BlackBoxSection}
with commuting unitary operators $U_1,\ldots,U_t$ on an $N$-dimensional complex vector space, and Assumption~\ref{defn:security-assumption} holds for the $U_j$, then the associated quantum lightning protocol is secure. 
Furthermore, the quantum lightning scheme is  
secure, without any hardness assumptions, against an adversary that makes only polynomially many black box calls to the $U_j$.
\end{thm}

\begin{proof}
The storm returns an answer in a reasonable amount of time. This is because for $z$ chosen randomly and $\delta < \eps/(10 t)$, for any joint eigenstate $\ket{\psi}$ there is at least a constant probability over the randomness of $z$ that none of the $\lambda_i+z$ have a real or imaginary part within $\delta$ of some multiple of $\eps/2$.

Since the verifier computes each $\lambda_i$ to error $\delta$, and since none of the $\lambda_i+z$ are within $\delta$ of a multiple of $\eps/2$ for a bolt produced by the storm, on such a bolt the verifier always returns the same serial number, and additionally the verifier returns distinct serial numbers for distinct bolts. Further, since measuring the eigenvalue of an eigenstate does not affect the state, the verifier only negligibly alters a bolt produced by the storm.

For security, 
suppose that a polynomial-time adversary, with probability $p$, produces two bolts that a verifier reads as having the same serial number. After the verifier is finished with them, the bolts will be in a state $\ket{\psi_1}\ket{\psi_2}\ket{\psi_3}\ket{\psi_4}$ for $\ket{\psi_i}$ some joint eigenstate of the $U_j$. Furthermore, to produce the same serial numbers (and not be rejected), 
the eigenvalues of $U_j$ on $\ket{\psi_i}$ and $\ket{\psi_k}$ must have real and imaginary parts differing by at most $\eps$ for all $i$ and $k$. By assumption, this implies that $\ket{\psi_i}=\ket{\psi_k}$ for all $i,k$. Throwing away the last register, this solves 
Problem~\ref{attackProblem}. If Assumption~\ref{defn:security-assumption} holds, then $p$ is negligible, proving security of the quantum lightning scheme.

By Corollary~\ref{cor:black-box-security}
the scheme is unconditionally 
secure against an adversary that only has black  box access to the $U_j$.
\end{proof}

Using Zhandry's derivation of  quantum money from quantum lightning, Theorem~\ref{thm51} can be used to give a quantum money system similar to the one described in \S\ref{BlackBoxSection}.

\section{Instantiation using Quaternion Algebras}
\label{sec:quatalgs}

Above, we discussed a quantum money protocol that depends on having access to a number of black box, commuting operators. However, for our protocol to be cryptographically secure, we will need to implement it using operators that are cryptographically difficult to work with. This is a bit of an issue as most easily computable sets of commuting operators will not be secure in this way. For example, taking $U_j$ to be the Pauli matrix on the $j$th qubit $Z_j$ gives an easy set of commuting operators, but
one for which it is easy to manufacture eigenstates (even
with specified eigenvalues). We come up with a hopefully secure set of commuting operators using the theory of quaternion algebras.

\subsection{Quaternion algebras}
\label{QuaternionSection}

Before we discuss our implementation in detail, we
review some basic facts about 
quaternion algebras over the field $\Q$ of rational numbers, for which
\cite{Voight} can be used as a reference.

\begin{defn}
Given non-zero $a,b \in \Q$, define $H(a,b)$ as the ring 
$$H(a,b) = \Q + \Q i + \Q j + \Q ij =
\{\alpha + \beta i + \gamma j +\delta ij : \alpha, \beta, \gamma, \delta \in\Q \}$$ 
with the 
relations $i^2 = a$, $j^2 = b$, and $ji = -ij$.
We define a {\em quaternion algebra over $\Q$} to be any such ring $H(a,b)$.
 \end{defn}
This definition of quaternion algebra over $\Q$ is not the standard one, but since every quaternion algebra over $\Q$ is an $H(a,b)$ for some $a$ and $b$, we take this as the definition. 
Note that  $H(a,b)$ has dimension four as a $\Q$-vector space, and
the Hamilton quaternions are $H(-1,-1)$.

For 
$z= \alpha + \beta i + \gamma j +\delta ij \in H(a,b)$ (with $\alpha, \beta, \gamma, \delta \in\Q$),
its conjugate is 
$\bar{z} := \alpha-\beta i-\gamma j-\delta ij$ and its reduced norm is $\nrd(z) := z\bar{z}$.

By definition, a {\em division algebra} is a ring in which every non-zero element has a multiplicative inverse.
A quaternion algebra $H$ over $\Q$ is {\em ramified} at a prime $N$ if the 
tensor product $H \otimes_\Q \Q_N$ of $H$ with the field of $N$-adic numbers $\Q_N$ is a division algebra (equivalently, is not the  
ring $M_2(\Q_N)$ of $2 \times 2$ matrices with entries in $\Q_N$).
We say $H$
is {\em ramified at $\infty$} if $H \otimes_\Q \R$ is a division algebra (equivalently, is not the ring $M_2(\R)$ of $2 \times 2$ real matrices).
An {\em order} $\co$ in $H$ is by definition a subring that is also a lattice (i.e., a finitely-generated $\Z$-submodule such that $\co \Q = H$).

From now on, suppose $N$ is a prime number and $N \ge 5$.
Let $H_N$ be the unique quaternion algebra over $\Q$ ramified at $N$ and $\infty$;
Proposition~5.1 of Pizer~\cite{Pizer} gives explicit $a$ and $b$ such that $H_N = H(a,b)$.
If $N \equiv 1 \pmod{6}$ let $\co_N$ be the maximal order given explicitly in Proposition~5.2 of Pizer~\cite{Pizer}, and if $N \equiv 5 \pmod{6}$ let
$H_N = H(-3,-N)$ and $\co_N = \Z + \Z\frac{1+i}{2} + \Z\frac{j+ij}{2} + \Z\frac{i-ij}{3}$.
(In particular, if $N \equiv 7 \pmod{12}$ then
$H_N = H(-1,-N)$ and $\co_N = \Z + \Z i + \Z\frac{1+j}{2} + \Z\frac{1+ij}{2}$.)
Then $\co_N$ is an $N$-extremal maximal order in $H_N$, that is, a maximal order for which the unique ideal of reduced norm $N$ is principal (see~\cite[Chapter 21]{Voight}).

A (left) fractional ideal of $\co_N$ is by definition a (full-rank) lattice in $H_N$ that is closed under left multiplication by elements of $\co_N$.
Define the ideal class set $\cl(\co_N)$ to be the set of fractional ideals of $\co_N$ modulo right multiplication, i.e., modulo the equivalence relation defined by $I\sim J$ if and only if there exists $z\in H_N$ such that $I=Jz$.
The ideal class set $\cl(\co_N)$ is finite (see
\cite[Chapter 21]{Voight}).
If $I$ is a fractional ideal of $\co_N$, let $[I]$ denote its ideal class, i.e., the set of fractional ideals $J$ of $\co_N$ such that $I=Jz$ for some $z\in H_N$.

If $I$ is a left fractional $\co_N$-ideal, then the reduced norm $\nrd(I)$ is defined in
\cite[\S 16.3]{Voight}, and satisfies $I \bar{I} = \nrd(I) \co_N$.
If $I \subset \co_N$ then $\nrd(I)^2 = [\co_N:I]$  (see \cite[\S 16.4.8]{Voight}).

The quaternion algebra $H(a,b)$ embeds in $\R^4$ via the homomorphism of abelian groups
$
\alpha + \beta i + \gamma j +\delta ij \mapsto (\alpha,\beta\sqrt{|a|},\gamma\sqrt{|b|},\delta\sqrt{|ab|}).
$
Identifying $H(a,b)$ with its image, for all $z \in H(a,b)$ we have  $\nrd(z) = \|z\|^2$, where $\|\cdot\|$ is the Euclidean norm on $\R^4$. The image of $\co_N$ and of any left fractional ideal of $\co_N$ is a lattice in $\R^4$.
We thus may represent a fractional ideal by a Minkowski reduced basis.
Since every fractional ideal is a rank four lattice, given a $\Z$-basis, a Minkowski reduced basis can be computed in polynomial time \cite{NguyenStehle}. In algorithms we specify a fractional ideal 
by a Minkowski reduced basis for it.

\subsection{Normalized Brandt operators $T(p)$}
\label{MpSection}
\begin{defn}
\label{RightOrderWtDef}
If $I$ is a left fractional $\co_N$-ideal, define its right order
$$\co_I:=\{z \in H_N : Iz \subset I\}$$ and its
weight $w_{I} := \#(\co_I^\times/\{\pm 1\})=\frac{1}{2}\#(\co_I^{\times})$.
\end{defn}
Then $\co_I$ is a maximal order, and $w_{I}$ depends only on the ideal class $[I]$.
In Proposition~\ref{prop:constants-w_i} in Appendix \ref{AppA} we completely describe the $w_{I}$. Our choices for the maximal orders $\co_N$ were
designed to give Proposition~\ref{prop:constants-w_i} a clean statement.

Suppose $p$ is a prime not equal to $N$, and
suppose $I$ and $J$ are non-zero fractional ideals of $\co_N$.
Define $a_p([I],[J])$ to be the number of fractional ideals $I'\subset J$ such that $I'\sim I$ and   $J/I' \cong \Z/p\Z \times \Z/p\Z$.
Let $h = \#\cl(\co_N)$ and let $T'(p)$ be the $h \times h$ matrix with $[I],[J]$-entry $a_p([I],[J])$. The matrix $T'(p)$
is the $p$-Brandt matrix for level $N$. The action of $T'(p)$ is self-adjoint for the pairing on $\C^{\cl(\co_N)}$ given by
$
  \left\langle (x_{[I]})_{[I]}, (y_{[I]})_{[I]}\right\rangle = \sum_{[I]} \frac{1}{w_{I}} x_{[I]} \overline{y_{[I]}}
$
(see \cite[\S 41.1.9]{Voight}).
Let $W$ be the diagonal $h \times h$ matrix whose $[I],[I]$-entry is $\sqrt{w_{I}}$, and let $T(p) = WT'(p)W^{-1}$. Then the $T(p)$ are real symmetric matrices that commute which each other (\cite[\S 41.1.10]{Voight}), and thus they have a simultaneous real eigenbasis. 
We call $T(p)$ the {\bf{normalized $p$-Brandt matrix}} for level $N$.
For example, if $N \equiv 1 \pmod{12}$, then
$T(p) = T'(p)$.

Let $V_N$ be the subset of $\C^{\cl(\co_N)}$ orthogonal (under the usual inner product) to the vector $(\sqrt{w_{I}})_{[I]}$.
Then $T(p)$ acts on $\C^{\cl(\co_N)}$, preserves $V_N$, and acts as a self-adjoint operator for the usual inner product.

In order to use the operators $T(p)$ in our quantum money scheme, we need to make them computationally tractable. First, we will need to find a better way of representing our ideal classes. While it is easy to give a single fractional ideal in the class, it is important for us
to find a canonical representation.

\subsection{Canonical encoding}
\label{CanonicalEncodingSection}

We next show how to obtain a canonical representation of an ideal class.

\begin{algorithm}\label{alg:can-representative}\leavevmode
  INPUT: A prime number $N \ge 5$, an $N$-extremal maximal order $\co_N$ in $H_N$, and a left fractional $\co_N$-ideal $I$.

  OUTPUT: A triple of integers $(d,a,b)$ such that $\gcd(d,a,b)=1$ and $b >a \ge 0$ and $d \geq 1$.
  \begin{enumerate}[(1)]
      \item Apply a shortest vector algorithm such as Algorithm 2.7.5 of \cite{hcohen} to the lattice $I$ to produce an element $z \in I$ of minimal non-zero reduced norm.
      \item Compute the ideal $J_z := \frac{1}{\nrd(I)}I\bar{z}$.
      \item Repeat steps (1) and (2) for each of the (at most six) $z \in I$ of minimal non-zero norm. Let $J$ be the ideal $J_z$ with lexicographically first encoding, and compute $m := \nrd(J)$.
      \item Compute the image $\ib \subset M_2(\Z/m\Z)$ of $J/m\co_N$ under the isomorphism
       $f_{N,m}: \co_N/m\co_N \isom M_2(\Z/m\Z)$ from the algorithm of Proposition~\ref{prop:algor-M2misomalgor} in Appendix \ref{AppA}.
      \item Letting $H \subset (\Z/m\Z)^2$ be the (cyclic) subgroup (of order $m$) generated by the rows of all the elements of $\ib$, apply the algorithm of Proposition~\ref{prop:algor-generatorm} to obtain $(d,c) \in \Z^2$ that generates $H$ and satisfies $d \mid m$ and $\gcd(d,c)=1$.
      \item Compute $b = m/d$ and $a = c \pmod{b}$. Output $(d,a,b)$.
  \end{enumerate}
\end{algorithm}
We call the triple $(d,a,b)$ obtained in this way the \emph{canonical encoding} of the ideal class of $I$. Theorem~\ref{thm:canrepthm} below justifies the terminology and shows that the algorithm works.

\begin{thm}
\label{thm:canrepthm}
In Algorithm~\ref{alg:can-representative}, we have:  \leavevmode
  \begin{enumerate}[(i)]
     \item $m\co_N \subset J \subset \co_N$; \label{can--2}
      \item $N \nmid m$; \label{can--1}
      \item $H$ is a cyclic group of order $m$; \label{can-1}
      \item $\gcd(d,a,b)=1$; \label{can-0}
      \item if inputs $I$ and $I'$ are in the same ideal class in $\cl(\co_N)$, then they output the same triple $(d,a,b)$, and produce the same $J$, $\ib$, and $H$; \label{can-3}
      \item if the same triple is output by inputs $I$ and $I'$, then  $[I] = [I']$,
      and $I$ and $I'$ produce the same $J$, $\ib$, and $H$; \label{can-4}
      \item Algorithm~\ref{alg:can-representative} is a quantum polynomial-time algorithm. \label{can-5}
\end{enumerate}
\end{thm}

\begin{proof}
If $\gamma \in H_N$ and $I_0$ is a left fractional $\co_N$-ideal, then
\begin{equation}\label{inversebareqn}
I_0\gamma \subset \co_N \text{ if and only if }
  \gamma \in I_0^{-1} = \overline{I_0}\nrd(I_0)^{-1}.
\end{equation}

Since $\bar{z}/\nrd(I) \in \bar{I}\nrd(I)^{-1} = I^{-1}$, it follows that
$J_z = I\bar{z}/\nrd(I) \subset \co_N$, so $J \subset \co_N$.
Then $1 \in J^{-1} = \bar{J}m^{-1}$ by \eqref{inversebareqn}, so $m\in J$, so $m\co_N \subset J$,
giving (i).

We claim that $m$ is the minimum of the reduced norms of the integral ideals in $[I]$. Say $I' = I\gamma$.
By \eqref{inversebareqn}, we have that
$I\gamma \subset \co_N$ if and only if
$\gamma = \bar{\alpha}/\nrd(I)$ with $\alpha \in I$.
The reduced norm
    $\nrd(I\gamma) = \frac{\nrd(\alpha)}{\nrd(I)}$
is minimized when $\alpha$ is an element of $I$ of minimal non-zero reduced norm. The minimality of $m$ follows.

Since $\co_N$ is $N$-extremal, the Frobenius ideal
of $\co_N$ is principal; let $\pi$ be a generator.
As in \cite[42.2.4]{Voight}, we have $J = \pi^rJ'$ for some $r\in \Z^{\ge 0}$ and some ideal $J' \subset \co_N$ satisfying $N \nmid \nrd(J')$. Then $m = \nrd(J) = N^r\nrd(J')$ and $J' \in [J] = [I]$.
By the minimality of $m$ we have $N \nmid m$, giving~\ref{can--1}.

If $r$ is a divisor of $m$, and $r \neq 1, m$, then $r\co_N \not\subset J$ and  $J \not \subset r\co_N$. To see this, first suppose $J \subset r\co_N$. Then $r^{-1}J$ is an integral ideal in the ideal class of $J$, of strictly smaller norm, contradicting the minimality of $m$. If $r\co_N \subset J$, then $r \in J$, so
by \eqref{inversebareqn} with $J$ in place of $I$, the ideal $J \bar{r}/m = Jr/m$ is integral.
It is then an integral ideal of strictly smaller norm in the ideal class of $J$, contradicting the minimality of $m$.
The map that sends a matrix to its rowspace induces a bijection from the set of left ideals of $M_2(\Z/m\Z)$
to the set of subgroups of $(\Z/m\Z)^2$ (Lemma \ref{lem:rowspacelem} in Appendix~\ref{AppA}).
It follows that
$
    r(\Z/m\Z)^2 \not\subset H$  and $H \not\subset r(\Z/m\Z)^2
$
for all non-trivial proper divisors $r$ of $m$,
from which one can show that the subgroup $H$ must be cyclic of order $m$, giving~\ref{can-1}.

Since $\gcd(d,c)=1$, we have $\gcd(d,c,b)=1$. Since $a \equiv c \pmod{b}$, we have~\ref{can-0}.

  For~\ref{can-3},
suppose that the inputs $I$ and $I'$ give $J$ and $J'$, respectively, in step (4) of the algorithm. Since $J'$ is an integral ideal in $[I]$ with minimal norm, as shown in the second paragraph of this proof there is an element $z \in I$ of minimal non-zero norm such that $J' = I \bar{\alpha}/\nrd(I)$.
Therefore when running the algorithm on input $I$, both $J$ and $J'$ appear in the list of ideals generated in step (3); by symmetry, the same occurs with input $I'$. Since both $J$ and $J'$ are lexicographically first, we have $J = J'$. Let $H$ be as in step (5).
By the last sentence of Proposition \ref{prop:algor-generatorm} in Appendix~\ref{AppA}, the integer $d$, and thus $b$, is uniquely determined. Suppose that $(d,c)$ and $(d,c')$ are two generators for $H$.
Then there exists $\lambda \in (\Z/m\Z)^\times$ such that $\lambda(d,c) = (d,c')$ in $H$. Since $\lambda d \equiv d \pmod{m}$, we have $\lambda \equiv 1 \pmod{b}$, so $c \equiv c' \pmod{b}$. Thus $a$ is also unique.

    For~\ref{can-4}, suppose that inputs $[I]$ and $[I']$ have the same output $(d,a,b)$.
The groups $H$ and $H'$ from step (5) of the algorithm are both subgroups of $(\Z/m\Z)^2$, where $m = db$.
The group $H$ is generated by $(d,c)$ for some $c$ with $a = c \pmod{b}$ and $\gcd(d,c)=1$, and $H'$ is generated by $(d,c')$ for some $c'$ with $a = c' \pmod{b}$ and $\gcd(d,c')=1$.
By Lemma \ref{lem:HJlem} in Appendix~\ref{AppA} we have $H = H'$.
Since (by Lemma \ref{lem:rowspacelem} in Appendix~\ref{AppA}) the map that sends a matrix to its rowspace induces a bijection from the set of left ideals of $M_2(\Z/m\Z)$
to the set of subgroups of $(\Z/m\Z)^2$, we have $\ib=\ib'$.
Then $J/m\co_N = J'/m\co_N$, so $J = J'$ and $[I] = [J] = [J'] = [I']$.

For \ref{can-5}, the $\Z$-rank of $I$ is $4$, so step (1) runs in polynomial time.

Viewed as lattices in $\R^4$, the index $[\co_N:I]$ can be computed as
the square root of
a ratio of
determinants. Since $\nrd(I) = \sqrt{[\co_N:I]}$, the reduced norm in step (2) can be computed in polynomial time.

In step (3), it is easy to compute all the elements of minimal non-zero norm from one of them, since each $I$ has at most six $z$ of minimal non-zero norm.  To see this, observe that $z, z' \in I$ both have minimal norm if and only if $z' = uz$ for some unit $u \in \co_N^\times$. The proof of Proposition~\ref{prop:constants-w_i} in Appendix~\ref{AppA} gives an explicit generator for $\co_N^\times$, which has order at most $6$. If $N \equiv 1 \pmod{12}$, then $\co_N^\times = \{\pm 1\}$, so up to sign there is a unique $z \in I$ of minimal non-zero norm, and only one ideal $J_z$.

Thus all steps run in polynomial time, except that the
invocation of Proposition~\ref{prop:algor-M2misomalgor} in step (4) might necessitate the use of a quantum polynomial-time
algorithm to factor $m$.
\end{proof}

Unfortunately, some triples $(d,a,b)$ are not canonical encodings, as seen in the following example. Fortunately, Algorithm \ref{alg:check-encoding} below enables one to detect when a triple is not canonical.

\begin{exa}
  Let $N = 23$, so $H_{23} = H(-3,-23)$ and $\co_{23}=\Z + \Z\frac{1+i}{2} + \Z\frac{j+ij}{2} + \Z\frac{i-ij}{3}$.
Let $\alpha = \frac{1+i}{2}$ and $\beta = \alpha + \frac{i-ij}{3} = \frac{3 + 5i -2ij}{6}$
and $I = (2,\beta)$. Then $\nrd(I) = 2$, and $I$, $I\alpha$, and $I\alpha^2$
are the only ideals in $[I]$ of minimal norm.
Applying the algorithm of Proposition~\ref{prop:algor-M2misomalgor}
 gives the isomorphism
  $\co_{23}/2\co_{23} \isom M_2(\Z/2\Z)$
that sends $\alpha$ to
  $\left[\begin{smallmatrix}
    0 & 1 \\
    1 & 1
  \end{smallmatrix}\right]$ and
  $\beta$
  to $\left[\begin{smallmatrix}
    0 & 0 \\
    1 & 1
  \end{smallmatrix}\right]$.
The image of $I$ (resp., $I\alpha, I\alpha^2$) is the set of matrices with row space generated by $(1,1)$
(resp., $(1,0), (0,1)$).
It follows that exactly one of $(1,1,2)$, $(1,0,2)$, and $(2,0,1)$
(depending on which of $I, I\alpha, I\alpha^2$ is lexicographically first) can be a canonical encoding of an ideal class in $\cl(\co_{23})$.
\end{exa}

\begin{algorithm}\label{alg:check-encoding} \leavevmode
  INPUT: A prime $N \ge 5$, a $\Z$-bases $(\omega_1,\omega_2,\omega_3,\omega_4)$ for a maximal order $\co_N$ in $H_N$, and a triple of integers $(d,a,b)$.

  OUTPUT: $1$ if $(d,a,b)$ is the canonical encoding of some fractional ideal of $\co_N$, along with an ideal $J \subset \co_N$ whose canonical encoding is $(d,a,b)$; $0$ otherwise.
  \begin{enumerate}[(1)]
      \item If $a \ge b$ or $a < 0$ or $d < 1$ or $b < 1$ or $\gcd(d,a,b) > 1$, output $0$ and stop.
      \item Apply Algorithm~\ref{lem:algor-liftinglem2} in Appendix~\ref{AppA} to compute an integer $c$ such that $\gcd(d,c) = 1$ and $c \equiv a \pmod{b}$.
      \item Set $m = db$. Apply the algorithm of Proposition~\ref{prop:algor-M2misomalgor} in Appendix~\ref{AppA} 
      to obtain an isomorphism
      $f_{N,m} : \co_N/m\co_N \isom M_2(\Z/m\Z),$
      let $\pi : \co_N \to \co_N/m\co_N \to M_2(\Z/m\Z)$ be the composition of reduction mod $N$ with $f_{N,m}$, and compute $\pi(\omega_i)$ for each $i$.
      \item Compute $x_i \in \Z$ such that $\sum_{i=1}^4 x_i \pi(\omega_i) = \left[
      \begin{smallmatrix}
        d & c \\
        0 & 0
      \end{smallmatrix}\right]$.
      \item Compute $\alpha = \sum_{i=1}^4 x_i \omega_i$ and compute a $\Z$-basis for the ideal $J \subset \co_N$ generated by $m$ and $\alpha$.
      \item
      Apply Algorithm~\ref{alg:can-representative} to compute the canonical encoding $(d',a',b')$ of $J$.
      \item Output $1$ and the ideal $J$ if $(d',a',b') = (d,a,b)$, and otherwise output $0$.
  \end{enumerate}
\end{algorithm}

\begin{prop}
  Algorithm~\ref{alg:check-encoding} is correct and runs in quantum polynomial time.
\end{prop}

\begin{proof}
Suppose that $I$ is a left fractional ideal of $\co_N$ and suppose that $(d,a,b)$ is its canonical encoding.
Let $c$, $m=db$, and
$J$ be as in Algorithm~\ref{alg:check-encoding} with input $(d,a,b)$.
To show correctness, by Theorem~\ref{thm:canrepthm}(v) it suffices to show that $[I]=[J]$.

Let $J'$,
$c'$, and $H' = \langle (d,c')\rangle$ be as in steps (3) and (5) of
Algorithm~\ref{alg:can-representative} with input $I$.
Let $H$ be the subgroup of $(\Z/m\Z)^2$ generated by $(d,c)$.
Then $\gcd(d,c)=1= \gcd(d,c')$ and $c \equiv c' \pmod{b}$.
By Lemma \ref{lem:HJlem} in Appendix~\ref{AppA} we have $H = H'$.
Since $J$ (resp., $J'$) is the inverse image, under the
composition $\co_N \to \co_N/m\co_N \isom M_2(\Z/m\Z)$, of the set of matrices whose rows are in $H=H'$, we have
 $J = J'$, so $[J] = [J']=[I]$.

Steps (4) and (5) are linear algebra. All steps run in polynomial time, except that steps (3) and (6) might necessitate the use of a quantum polynomial-time
algorithm to factor $m$.
  \end{proof}

  We will need to bound the size of $m$ that we may need to deal with.

\begin{lem}
\label{MinkLem}
Suppose $z$ is an element of minimal non-zero norm in a fractional ideal $I$ for $\co_N$. Let
$J = \frac{1}{\nrd(I)}I\bar{z}$ and $m = \nrd(J)$.
Then $m \le \sqrt{2}\sqrt{N}$.
\end{lem}

\begin{proof}
Let $\lambda_1(I)$ denote the length of a shortest non-zero vector in the lattice $I$,
and let $D$ denote the discriminant of $I$.
By the Hermite bound we have
$\lambda_1(I)^4 \le 2|\det(I)| = 2\sqrt{D}$.
But $\nrd(z) = \lambda_1(I)^2$, so $\nrd(z) \leq \sqrt{2} \sqrt[4]{D}$.

By Lemma 15.2.15 and Proposition~16.4.3 of \cite{Voight} we have
$\nrd(I) = \sqrt[4]{D}/\sqrt{N}$. Thus
$m = {\nrd(z)}/{\nrd(I)}
    = \nrd(z) {\sqrt{N}}/{\sqrt[4]{D}}
    \leq \sqrt{2}\sqrt{N}$.
  \end{proof}

\subsection{Computation of normalized Brandt operators $T(p)$}
\label{BrandtMatrixComputationSect}
Given an ideal class $[J]$,
we will need to find the multiset of ideal classes $[I]$ with non-zero $a_p([I],[J])$-entries. This is relatively straightforward as we need to find $I\supset J \supset pI$ that are invariant under left multiplication
by $\co_N$, or equivalently we need to find $J/pI \subset I/pI$ that are invariant under $\co_N/p\co_N$. It is a standard fact that the action of $\co_N/p\co_N$ on $I/pI$
is isomorphic to the action of $M_2(\Z/p\Z)$ on itself.
Once such isomorphisms are computed using the algorithm of Proposition~\ref{prop:algor-M2misomalgor} in Appendix~\ref{AppA},
the invariant elements of $I/pI$
correspond to $\{A\in M_2(\Z/p\Z) \mid Av=0\}$ for $v$ some non-zero element in $(\Z/p\Z)^2$. Since these sets are invariant under scaling of $v$, there are exactly $p+1$ such $J$'s,
and they are computable in a straightforward manner. Furthermore, since $J$ is a small index sublattice of $I$, from a reduced basis of $I$  it is relatively simple to compute a reduced basis for $J$ and thus, the appropriate canonical representation for $[J]$.
This allows us to compute the non-zero entries of a row of the Brandt matrix $T'(p)$. Proposition~\ref{prop:constants-w_i} in Appendix~\ref{AppA} gives the $w_{I}$, and hence the normalized Brandt matrix $T(p)$. Then, using standard Hamiltonian simulation algorithms, it is straightforward to approximate the action of $e^{iT(p)/\sqrt{p}}$ on $V_N$. By~\cite[Theorem 1]{hamsim}, $e^{iT(p)/\sqrt{p}}$ can be computed with gate complexity polynomial in $p$ and $\log(N)$.

If $p$ is small compared to $N$, then $T(p)$ is a sparse matrix,
since each column has at most $p+1$ non-zero entries and the matrix is $h \times h$  with $h = O(N/12)$.

\subsection{Producing maximally entangled states}
\label{bellSection}

There is one additional difficulty in implementing our scheme in this context. Namely, there is no obvious way to produce a maximally entangled state for $V_N \otimes V_N$. In this section, we provide an efficient algorithm for doing this.

To produce this representation, first note that it suffices to produce a state
that is a uniform superposition of the representatives for the elements of $\cl(\co_N)$. To do this, we begin by providing a different representation of such elements.

The following algorithm efficiently produces a superposition over canonical encodings of ideal classes.
\begin{algorithm}\label{alg:bell-stat} \leavevmode
    INPUT: A prime number $N \ge 5$ and an $N$-extremal maximal order $\co_N$.

  OUTPUT: Either a quantum state proportional to $\sum \ket{d,a,b}$, where
$(d,a,b)$ varies over the canonical encodings of the elements in $\cl(\co_N)$, or else $\perp$.
\begin{enumerate}[(1)]
    \item Prepare a state $\ket{\psi}$ proportional to $\sum_{d=1}^{\lfloor \sqrt{2N} \rfloor} \frac{1}{\sqrt{d}}\ket{d}.$
    \item Apply to $\ket{\psi}$ the linear map that sends $\ket{d}$ to
$
\ket{d}\otimes \Bigl(\frac{1}{\sqrt{C_d}}\sum_{i=1}^{C_d}\ket{i} \Bigr)^{\otimes 2},
$
where $C_d = \lfloor \sqrt{2N}/d \rfloor$,
and call the resulting state $\ket{\psi_1}$.
  \item Writing $f$ for a function that implements Algorithm~\ref{alg:check-encoding}, apply to
 $\ket{\psi_1}\ket{0}$ the operator that sends
$\ket{d,a,b}\ket{0}$ to $\ket{d,a,b}\ket{f(d,a,b)}$,
and call the resulting state $\ket{\psi_2}$.
  \item Measure the last register of the quantum state $\ket{\psi_2}$. If the result is $0$, output $\perp$. Otherwise, output $\ket{\psi_2}$.
\end{enumerate}
\end{algorithm}
If the algorithm outputs $\perp$, we say it fails; otherwise we say it succeeds.

\begin{thm}\label{thm:bell-state-prob}
  Algorithm~\ref{alg:bell-stat} succeeds with probability at least $(1-\frac{1}{N})\frac{1}{32 \pi^2}$.
\end{thm}

\begin{proof}
Step (1) can be implemented by first preparing the state
    $\frac{1}{\sqrt{M}} \sum_{m=1}^M \ket{m}$
 where $M := \lfloor \log_2(\sqrt{2N}) \rfloor$, then
appending a zero qubit to this state and applying the operator defined by
  \[
    \ket{m0} \mapsto \sum_{d = 2^m}^{2^{m+1}} \Biggl[\frac{1}{\sqrt{d}}\ket{d0} + \sqrt{\frac{1}{2^m} - \frac{1}{d}}\ket{d1}\Biggr],
  \]
and then measuring the last qubit. If the result is $1$, start over. If the result is $0$, step (1) has succeeded. One can compute that the success probability for step (1) is at least $\frac{1}{2}$.

For each $d$ let $C'_d$ be the least power of $2$ larger than $C_d$.
To implement step (2), consider the following procedure. First, apply the operator defined by
    $\ket{d0} \mapsto \frac{1}{\sqrt{C_d'}}\ket{d} \sum_{i=1}^{C_d'} \ket{i}.$
Define $B(d,i) := 0$ if $i > C_d$ and $B(d,i) := 1$ if $i \leq C_d$. Apply the operator defined by
    $\ket{di} \mapsto \ket{di}\ket{B(d,i)}$
  and measure the last register. If the result is $0$, reject and start over. Rejection occurs with probability $\leq \frac{1}{2}$. If the result is $1$, discarding the last register leaves a state of the form
 $\frac{1}{\sqrt{C_d}}\ket{d} \sum_{i=1}^{C_d} \ket{i}.$
  Applying the above procedure twice produces the output of step (2). The success probability for step (2) is $ \ge \frac{1}{4}$.

Step (2) outputs a state that approximates the uniform distribution of $\ket{d,a,b}$ for $(d,a,b)$ running over positive integers with $da,db \leq \sqrt{2N}$.
By Lemma \ref{MinkLem}, these triples include all the canonical encodings of elements
of $\cl(\co_N)$.
  The number of states in this distribution is thus
    $$\sum_{d=1}^{\lfloor \sqrt{2N} \rfloor} \Biggl(\frac{\sqrt{2N}}{d}\Biggr)^2 \le 2N \sum_{d=1}^\infty \frac{1}{d^2}
    = \frac{\pi^2N}{3}.$$
  The number of triples $(d,a,b)$ that are canonical encodings is  $\#\cl(\co_N) \ge (N-1)/12$. Thus the success probability in step (4) is at least $((N-1)/12)/(\pi^2N/3) = (1 - \frac{1}{N})\frac{1}{4\pi^2}$. The claim follows.
  \end{proof}

  Using Algorithm~\ref{alg:bell-stat}, we can next obtain a maximally entangled state by applying a controlled-NOT operator. We implement this in the next section in our instantiation of the minting algorithm.

\subsection{Instantiation of protocol}
\label{sec:inst-prot}
Algorithm \ref{MintAlgInstant} below is our instantiation of the minting algorithm using normalized Brandt operators.
    Fix $N$, $H_N, \co_N$, and $V_N$ as before, and choose primes $p_1,\ldots,p_t$ distinct from $N$. Let $T(p_j)$ be the normalized $p_j$-Brandt matrix for level $N$ (as defined in \S\ref{MpSection}) and let $U_{j} = e^{i T(p_j)/\sqrt{p_j}}$.
Recall that \S\ref{BrandtMatrixComputationSect} allows one to compute $T(p_j)$ and the action of $U_j$ on $V_N$.
Let $\{\basis{i}\}$ be a simultaneous real eigenbasis for the $U_j$'s. Fix $\eps$ for which $\{\basis{i}\}$ is $\eps$-separated; based on empirical evidence, $\eps = \frac{1}{4(\log_2N)}$ should be a suitable choice, but see \S{}\ref{sec:eps-separation} for additional discussion. Set the public parameters $\pp$ to be $(N,H_N, \co_N, p_1,\ldots,p_t,\eps)$. Let $\dsk$ be a signing key for a fixed digital signature algorithm. The $\mint$ algorithm is as follows.
\begin{algorithm}
\label{MintAlgInstant}
  INPUT: $\pp, \dsk$

  OUTPUT: a uniformly random valid bill $(\ket{\psi},v,\sigma)$ with $\ket{\psi} \in V_N$.
\begin{enumerate}[(1)]
    \item Apply Algorithm~\ref{alg:bell-stat} to obtain the superposition $\sum \ket{d,a,b}$, where the sum is over all the canonical encodings of the elements of $\cl (\co_N)$. \label{bell-step}
    \item Append an ancillary register initialized to $0$ and apply controlled-NOT operators to obtain the quantum state $\sum (\ket{d,a,b} \otimes \ket{d,a,b})$.
    \item Apply phase estimation with the operators $U_j \otimes I_h$ and $I_h \otimes U_j$ for $j = 1,\ldots,t$.
Let $\ket{\psi}$ be the resulting quantum state and $v$ the tuple of eigenvalues.
    \item If $v_j = e^{i(p_j+1)/\sqrt{p_j}}$ for all $j$, output $\perp$. \label{check-bad-eigenvector}
    Otherwise, let $\sigma$ be the signature of $v$, and output $(\ket{\psi},v,\sigma)$.
\end{enumerate}
\end{algorithm}

This instantiation of the minting algorithm is essentially the same as the black box $\mint$ algorithm of \S\ref{BlackBoxSection}, except that it takes place in the subspace $V_N$, not in all of $\C^{\cl (\co_N)}$.

\begin{thm}
The $\mint$ Algorithm \ref{MintAlgInstant} has correct output with probability at least $\frac{1}{32 \pi^2} (1-\frac{1}{N})(1 - \frac{12}{N})$, and runs in quantum polynomial time.
\end{thm}

\begin{proof}
   The subspace $V_N$ is the orthogonal complement of the span of
      $\ket{\psi_0} := \sum \frac{1}{w_{I}}\ket{d,a,b},$
where the sum is over the canonical encodings of the elements of $\cl (\co_N)$,
and $w_{I}=\#\co_I^\times/\{\pm 1\}$ where $\co_I$ is the right order of any ideal $I$ with canonical encoding $(d,a,b)$.
 As $\basis{0}$ is itself an eigenvector for the normalized Brandt operators, in the minting protocol it suffices to check that the output is not $\basis{0} \otimes \basis{0}$. (When $N \equiv 1 \pmod{12}$, then $\basis{0} \otimes \basis{0}$ is the maximally entangled state as in Step 2 of Algorithm~\ref{MintAlgInstant}.) Since $T(p)\ket{\psi_0} = (p+1)\ket{\psi_0}$, we have $U_p \ket{\psi_0} = e^{i(p+1)/\sqrt{p}} \ket{\psi_0}$. By $\eps$-separation, step~\ref{check-bad-eigenvector} outputs $\perp$ if and only if $\ket{\psi} = \basis{0} \otimes \basis{0}$; otherwise, the output is a valid bill with a note in $V_N$.

  Step~\ref{bell-step} succeeds with probability bounded below by
  $\frac{1}{32 \pi^2}(1-\frac{1}{N})$, by Theorem~\ref{thm:bell-state-prob}. The state proportional to
$
    \sum (\ket{d,a,b} \otimes \ket{d,a,b})
$
  is a uniform superposition of all of the eigenstates of the form $\basis{i} \otimes \basis{i}$. Since there are at least $N/12$ such states, the probability of obtaining the state $\basis{0} \otimes \basis{0}$ is at most $\frac{12}{N}$. The claim follows.
\end{proof}

\begin{rmk}
  Our motivation for working in $V_N$ instead of $\C^{\cl \co_N}$ is that $\basis{0} \otimes \basis{0}$ is an easy state to manufacture, so allowing this state to be a valid note would permit easy attacks by the mint or by others as in \S{}\ref{sec:blackboxsecprob}.
\end{rmk}

The instantiation of the verification algorithm $\ver$ is identical to the black box algorithm of \S{}\ref{BlackBoxSection}.

\subsection{$\eps$-separation}
\label{sec:eps-separation}

Our quantum money protocol instantiation requires that the eigenbasis for the operators $e^{iT(p)/\sqrt{p}}$
be $\epsilon$-separated. Table~\ref{tab:epsilon-separation-n} in Appendix~\ref{AppB} gives experimental data that suggests that the eigenbasis is $\epsilon$-separated even for $\epsilon$ quite large; for instance,  $\epsilon = 1/(4\log_2(N))$ works for all $N$ in Table~\ref{tab:epsilon-separation-n}, where we use $e^{iT(p)/\sqrt{p}}$ for all primes $p < \log_2(N)$.

For the normalized Brandt operators $T(p)$, rather than $e^{iT(p)/\sqrt{p}}$,
Goldfeld and Hoffstein \cite{GoldfeldHoffstein} obtain $\epsilon$-separation when the number of operators $m$ is $O(N \log N)$, and obtain a bound for $m$ that {is} $\polylog(N)$ if they assume a version of the Riemann hypothesis. However, the $\epsilon$ is not explicit. Note that Goldfeld and Hoffstein, as well as Serre's result below, deal with the Brandt operators $T'(p)$, but since $T(p)$ and $T'(p)$ are similar, the eigenvalues are identical, and so these results apply to $T(p)$ as well.
\begin{thm}[\cite{GoldfeldHoffstein}, Theorems 3 and 2]
\label{GoldfeldHoffsteinThm}
  Let $N \ge 5$ be a prime. For each prime $p \ne N$, let $T(p)$ be the normalized $p$-Brandt matrix for level $N$. Then:
  \begin{enumerate}[(i)]
      \item
   There exist a constant $K = O(N \log N)$ and $\epsilon > 0$ such that if $p_1,\ldots,p_t$ is the list of primes $\leq K$ with $p_i \neq N$, then every eigenbasis for the operators $T(p_1), \ldots, T(p_t)$ is $\epsilon$-separated.
     \item
 If the
 Riemann hypothesis for Rankin-Selberg
zeta functions holds, then there exist a constant $K = O((\log N)^2(\log \log N)^4)$ and $\epsilon > 0$ such that if $p_1,\ldots,p_t$ is the list of primes $\leq K$ with $p_i \neq N$, then every eigenbasis for the operators $T(p_1), \ldots, T(p_t)$ is $\epsilon$-separated. If $N > e^{15}$, one can take
    $K = 16(\log N)^2(\log \log N)^4.$
  \end{enumerate}
\end{thm}

\begin{prop}\label{lemma:eps-prime-bis}
Suppose $p_1\ldots,p_t, N$ are distinct prime numbers with $N \ge 5$.
Then every eigenbasis with respect to $T(p_1),\ldots,T(p_t)$ that is $\epsilon$-separated for some $\epsilon > 0$ is also an eigenbasis with respect to $\big\{e^{\frac{i}{\sqrt{p_j}}T(p_j)}\big\}$ that is $\epsilon'$-separated for some $\epsilon'$ such that $0 <\epsilon' =  O(\epsilon/\sqrt{\max_j\{p_j\}})$.
\end{prop}

\begin{proof}
Let  $U_j = e^{\frac{i}{\sqrt{p_j}}T(p_j)}$ and $K = \max_j\{p_j\}$.
By Deligne's proof of the Weil conjectures~\cite{deligne},
the eigenvalues of $T(p_j)$ lie in the interval $[-2\sqrt{p_j},2\sqrt{p_j}]$, so the eigenvalues of $\frac{1}{\sqrt{p_j}} T(p_j)$ lie in $[-{2},{2}]$.  Let $H = \{z \in S^1 \mid -2 \leq \mathrm{arg}(z) \leq 2\}$. The map on $t$-tuples $\rho: [-{2},{2}]^t \to H^t$ given coordinate-wise by $\lambda \mapsto e^{i\lambda}$ sends tuples of eigenvalues 
for the $\frac{1}{\sqrt{p_j}}T(p_j)$ to the corresponding tuple of eigenvalues 
for the $U_j$. Since $\rho^{-1}$ is Lipschitz continuous, there exists $M > 0$ such that for all $v_1, v_2 \in [-2,2]^t$ we have
  $|v_1 - v_2| \leq M|\rho(v_1) - \rho(v_2)|$.
If $v_1$ and $v_2$ are two distinct tuples of eigenvalues for $\frac{1}{\sqrt{p_j}}T(p_j)$, then $|v_1 - v_2| > \frac{1}{\sqrt{p_j}} \epsilon$. It follows that with respect to the $U_j$'s our joint eigenbasis is $\epsilon'$-separated for  $\epsilon' = \frac{1}{M\sqrt{K}}\epsilon$.
  \end{proof}

Fix primes $p_1, \ldots p_t$.
As before, for each prime $N$ distinct from $p_1,\ldots,p_t$,
let $\{v_{i,N}\}_{i=1}^h$ be the set of vectors of eigenvalues for an eigenbasis for $\{{\frac{1}{\sqrt{p_j}}T(p_j)}\}_{j=1}^t$, where $T(p_j)$ is the normalized $p_j$-Brandt matrix for level $N$ and $h=\#\cl(\co_N)$. On the interval $[-2,2]$, let $\mu_p$ denote the probability measure
$
  \frac{p+1}{\pi} \cdot \frac{(1-{x^2}/{4})^{1/2}}{(\sqrt{p}+\frac{1}{\sqrt{p}})^2 - {x^2}} dx.
$
\begin{thm}[Th\'eor\`eme 3, \cite{SerreDistribution}]
\label{SerreDistributionThm}
The distribution of vectors $\{ v_{i,N} \}_{i=1}^h \subset [-2,2]^t$, where $N$ is a prime not equal to $p_1,\ldots,p_t$,  approaches the product measure $\prod_{i=1}^t \mu_{p_i}$ as $N$ goes to infinity.
\end{thm}
For $p$ large, $\mu_p$ approaches the distribution
  $\frac{1}{2\pi}\sqrt{4-x^2} dx$.
Thus the distribution of the eigenvalues of the $U_j = e^{iT(p_j)/\sqrt{p_j}}$ in the subset of $S^1$ with argument $x \in [-2,2]$ will approach the distribution $\frac{1}{2\pi}\sqrt{4-x^2} dx$. A more precise statement on the distribution of eigenvalues is given in \cite[Theorem 19]{murty-sinha}.

\begin{rmk}
\label{remk616}
In light of Theorem~\ref{SerreDistributionThm}, a natural assumption is that the $v_{i,N}$ act like independent random samples drawn from the distribution $\prod_{i=1}^t \mu_{p_i}$. Under this assumption,
if $0 < \epsilon < 1$, then for $t$ larger than a sufficiently large multiple of $\log N$, with high probability the eigenbasis for the $U_j$ is $\epsilon$-separated.
An open question is how the eigenvectors of the Brandt matrices vary. If they also act like independent random samples as $N$ varies, then the operators $U_j$ also act like random (commuting) unitary operators.
\end{rmk}

\section{Security of the Instantiation}
\label{attacksSection}

As the operators $U_j$ in the instantiation are no longer black box, one must now consider  additional attacks. In \S\S\ref{OtherUjss}--\ref{ModularFormsSection} we note some of the most obvious attacks on Problem~\ref{quat-Problem}, and reasons we do not expect them to work. In each case, instead of an attacker with only black box access to the $U_j$ we consider an attacker that uses some property of the instantiation.

\subsection{Security reduction}
\label{reductionss}
The following problem restates Problem~\ref{attackProblem} in the setting of our instantiation.
\begin{problem}
\label{quat-Problem}
Given a prime $N \ge 5$, and operators $U_j = e^{iT(p_j)/\sqrt{p_j}}$, where  the $T(p_j)$ are the normalized Brandt matrices acting on $V_N$ corresponding to distinct primes $p_1, \ldots, p_t$ not equal to $N$, output a state of the form $\ket{\psi}\ket{\psi}\ket{\psi}$, where $\ket{\psi}$ is an eigenvector for all the $U_j$'s.
\end{problem}

The proof of Theorem~\ref{thm:forgery-attackproblem} shows that 
our instantiation is secure if Problem~\ref{quat-Problem}
is hard and the digital signature algorithm is secure.

\subsection{Use of other $U_j$}
\label{OtherUjss}
An attacker will have access not just to the $U_j$ used in the quantum money protocol but also to
$e^{iT(p)/\sqrt{p}}$ for other primes $p$. Since the black box lower bound from Theorem~\ref{LowerBoundTheorem} does not depend on the number of operators, its conclusion still holds, if one were to treat this larger set of operators as black box operators.

\subsection{Other powers of $e^{i T(p)/\sqrt{p}}$}
An attacker 
can apply arbitrary powers of the $U_j$,
by computing $e^{i \gamma T(p)/\sqrt{p}}$ for any $\gamma \in \R$.
The following modification of Theorem~\ref{LowerBoundTheorem} shows 
this does not help.

\begin{thm}\label{LowerBoundTheorem-powers}
  Suppose $\mathcal{A}$ is an algorithm that, on input a real number $\gamma \in \R$ and a black-box unitary operator $e^{T}$, outputs a black-box unitary operator that approximates $e^{\gamma T}$.
      Suppose $\mathcal{D}$ is any probability distribution over $(S^1)^t$ such that with high probability, any finite number of samples chosen from $\mathcal{D}$ are distinct. Then any circuit consisting of standard gates and controlled $\mathcal{A}(\gamma,U_j)$ gates that solves Problem~\ref{attackProblem} with constant positive probability for sets of operators $U_1,\ldots,U_t$ chosen according to $\mathcal{D}$ and with uniformly random real eigenbasis $\{\basis{i}\}$ must have
  $\querycomp$ controlled $\mathcal{A}(\gamma,U_j)$ gates.
\end{thm}
The only difference between the proofs of Theorems~\ref{LowerBoundTheorem-powers} and \ref{LowerBoundTheorem} is that the calls to $\mathcal{A}$ might give a different distribution $\mathcal{D}$ of eigenvalues. We may assume that the $\gamma$ chosen in Theorem~\ref{LowerBoundTheorem-powers} all satisfy $0 < \gamma \leq 1$. Then the distribution induced by replacing each sample from $\mathcal{D}$ with its $\gamma$th power for some $0 < \gamma \leq 1$ also has the property that with high probability, any finite number of samples are distinct.

\subsection{Sparse logarithms}
\label{sec:sparse-logarithms}
The matrices $T(p)$ are too large to be able to directly compute their eigenvectors via classical algorithms from linear algebra.
However, the $\log(U_j) = \frac{1}{\sqrt{p_j}}T(p_j)$ used in our protocol are sparse operators.
One could ask whether one could
use an HHL-like quantum algorithm \cite{HHLpaper} to find eigenvectors (one cannot use HHL directly as the matrix used would not be invertible).
Since an HHL-like algorithm deals with $e^{itT(p)}$ for $t \in \R$ via Hamiltonian simulation,  rather than directly with the sparse matrices $\frac{1}{\sqrt{p}}T(p)$, 
security against such attacks is covered
by our black box lower bounds in Theorem~\ref{LowerBoundTheorem}.

\subsection{Quantum state restoration}

A technique in \cite{staterestoration} was developed to break a number of quantum money schemes that look superficially like ours. These schemes use eigenstates of some operator $H$ where the state itself has some clean (but secret) product representation.
In \cite{staterestoration} it is shown that if we are
given a state $\ket{\psi}=\ket{\psi_A}\otimes \ket{\psi_B}\in V_A\otimes V_B$ and can compute a measurement of whether we are in state $\ket{\psi}$, we can produce a duplicate of the state $\ket{\psi_B}$ in time $\poly(\dim(V_B))$. If the supposedly secure state is a tensor product of many small pieces, this can be used to recover the individual pieces one at a time.

We argue that it is extremely unlikely that the eigenstates in our algorithm can be decomposed as such tensor products. In fact, it is extremely unlikely that there is even any natural way to write $V_N$ as a tensor product. For suppose $V_N \cong W_1 \otimes W_2$. For each Brandt operator $T(p)$, the eigenvalues of $T(p)$ acting on $V_N$ would be products $t_{ij} = \lambda_{i} \rho_{j}$, where the $\lambda_i$ (respectively $\rho_j$) are the eigenvalues of $T(p)$ acting on $W_1$ (respectively $W_2$). These eigenvalues would then satisfy quadratic relations $t_{ij}t_{k\ell} = t_{kj}t_{i\ell}$.
 Theorem~\ref{SerreDistributionThm} (due to Serre) suggests that the eigenvalues of the $T(p)$ act like random variables taken from the distribution $\mu_p$
and thus that
the eigenvalues 
satisfy the above quadratic relations with probability zero.

\subsection{Modular forms and elliptic curves}
\label{ModularFormsSection}

As mentioned in \S{}\ref{sec:motiv-proposed}, there is a well-understood connection between Brandt operators and both modular forms and supersingular elliptic curves.
Since the most efficient way known to compute modular forms is to compute the eigenvalues of Brandt operators (see \cite{SageBrandt,Pizer,Voight,Kohel}),
it seems unlikely that one could use modular forms
to attack a protocol based on quaternion algebras. We make this more precise below.
We focus on modular forms, but could equivalently
phrase it using elliptic curves, via the equivalence given in \cite{Mestre}.
A reference for modular forms is \cite{LangModForms}.

To try to solve Problem~\ref{quat-Problem}, one could try to directly manufacture a specific eigenstate $\ket{\psi}$ three times in succession to obtain a solution to Problem~\ref{quat-Problem}.
We next consider two ``direct manufacture'' problems.

Cusp forms are typically encoded as power series $f(q) = \sum_{n=1}^\infty a_n q^n$. We assume that for $\polylog(N)$-many primes $p$, $\polylog(N)$-many bits of the coefficient $a_p$ are specified. If $f$ is a simultaneous eigenvector of all the
Hecke operators $T_p$, normalized so that $a_1 = 1$, then the eigenvalue of $T_p$ is $a_p$. Such a cusp form $f$ is called an {\em eigenform} for the Hecke operators.
\begin{problem}\label{prob:cusp-form-state}
  Given a prime $N$  and a normalized eigenform $f \in S_2(\Gamma_0(N))$ for the Hecke operators $T_p$ with corresponding eigenvalue $a_p$ for all primes $p$,
output a simultaneous eigenstate $\ket{\psi}$ of $e^{iT(p)/\sqrt{p}}$ for all primes $p$, such that the corresponding eigenvalue is $e^{ia_p/\sqrt{p}}$.
\end{problem}

\begin{problem}\label{prob:eigenvalue-state}
  Given a prime $N \ge 5$, complex numbers $\alpha_1,\ldots, \alpha_t$, operators $U_j = e^{iT(p_j)/\sqrt{p_j}}$, where the $T(p_j)$ are the normalized Brandt matrices acting on $V_N$ corresponding to distinct primes $p_1, \ldots, p_t$ not equal to $N$, and a promise that there is a simultaneous eigenstate of $U_1,\ldots,U_t$ such that $U_j$ has
eigenvalue $\alpha_j$ for each $j$, output such a simultaneous eigenstate $\ket{\psi}$ of $U_1,\ldots,U_t$ with eigenvalues $\alpha_1,\ldots, \alpha_t$, respectively.
\end{problem}

No efficient algorithms are known to solve Problems~\ref{prob:cusp-form-state} or \ref{prob:eigenvalue-state}.

\begin{lem}\leavevmode \label{lem:unique-soln-cusp}
  \begin{enumerate}[(i)]
      \item Every solution to Problem~\ref{prob:cusp-form-state} is unique
    (up to scalar).
      \item  If every instance of
    Problem~\ref{prob:eigenvalue-state} with fixed choice of $N$ and $U_1,\ldots,U_t$ as part of the input has a solution that is unique up to scalar, then every eigenbasis for this choice of  $N, U_1,\ldots,U_t$ is $\epsilon$-separated for some $\epsilon > 0$.
       \item Given $N$ and $U_1,\ldots,U_t$, if there
    is an $\epsilon$-separated eigenbasis for the $U_j$'s for some $\epsilon > 0$,
    then every solution to
    Problem~\ref{prob:eigenvalue-state} with these $N$ and $U_1,\ldots,U_t$ as part of the input is unique (up to scalar).
\end{enumerate}
\end{lem}
\begin{proof}
For (i), fix an instance $N$ and $f = \sum a_nq^n$ of Problem~\ref{prob:cusp-form-state}, and suppose $\ket{\psi}$ and $\ket{\psi'}$ are simultaneous eigenstates for $e^{iT(p)/\sqrt{p}}$ with eigenvalue $e^{ia_p/\sqrt{p}}$ for all primes $p$. Then for each prime $p$, the states $\ket{\psi}$ and $\ket{\psi'}$ are simultaneous eigenvectors for the operators $T(p)$. The eigenvalue for $T(p)$ of $\ket{\psi}$ is $\frac{a_p}{p_j} + 2\pi k_p$ for some $k_p\in\Z$. Since $T(p)$ is an integer matrix, its eigenvalues, including $a_p$, are algebraic numbers. Therefore $2\pi k_p$ is also algebraic, so $k_p = 0$ for all $p$. Thus, $a_p$ is the eigenvalue for the operator $T(p)$ of $\ket{\psi}$, and similarly of $\ket{\psi'}$.
By the multiplicity one theorem for weight two cusp forms of prime level,
two normalized eigenforms in $S_2(\Gamma_0(N))$ with the same eigenvalues for all the Hecke operators $T_p$ must be equal.
Since the system of Hecke operators $T_p$ acting on $S_2(\Gamma_0(N))$
is isomorphic to the system of operators $T(p)$ acting on $V_N$, it follows that $\ket{\psi}$ and $\ket{\psi'}$ are scalar multiples, giving (i).

  For (ii), suppose $\{\basis{i}\}_{i=1}^h$ is an eigenbasis that is not $\epsilon$-separated for any $\epsilon$, with eigenvalues $z_{ij}$ satisfying $U_j\basis{i} = z_{ij}\basis{i}$ for $i=1,\ldots,h$ and $j=1,\ldots,t$.
Then there exist $k \neq \ell$ such that $z_{kj} = z_{\ell j}$ for all $j$. Set $\alpha_j = z_{kj}$ for each $j$. Then $\basis{k}$ and $\basis{\ell}$ are linearly independent solutions to Problem~\ref{prob:eigenvalue-state}, for the given $N,U_1,\ldots,U_t$. This gives (ii).

  For (iii), suppose
 $\{\basis{i}\}_{i=1}^h$ is an $\epsilon$-separated eigenbasis  for some $\epsilon>0$, with eigenvalues $z_{ij}$ satisfying $U_j\basis{i} = z_{ij}\basis{i}$,
 and suppose
 $\ket{\psi}$ and $\ket{\psi'}$ are two solutions to Problem~\ref{prob:eigenvalue-state}.
Write $\ket{\psi} = \sum c_i \basis{i}$ and $\ket{\psi'} = \sum c_i' \basis{i}$ with $c_i,c_i' \in\C$.
Applying $U_j$ to both equations gives
    $c_i \alpha_j = c_i z_{ij}$
    and
    $c_i' \alpha_j = c_i'z_{ij}$
    for all $i$ and $j$.
  Choose $k$ so that $c_k \neq 0$. Then $\alpha_j = z_{kj}$ for all $j$.
Suppose $i \neq k$.
Since $\{\basis{i}\}$ is $\epsilon$-separated, there exists $j$ such that $z_{ij} \neq z_{kj} = \alpha_j$. It follows that $c_i = c_i' = 0$ for all $i \neq k$. Thus both $\ket{\psi}$ and $\ket{\psi'}$ are non-zero multiples of $\basis{k}$,  giving (iii).
\end{proof}

In the next result, ``with complexity $T$'' for classical algorithms means in time $T$, and for quantum algorithms means with gate complexity $T$.

\begin{prop}
\label{propBCalgors}
Suppose there are an algorithm $B$ that on input $N$ can solve Problem~\ref{prob:eigenvalue-state} with complexity $B_N$, and
an algorithm $C$ that on input $N$ with complexity $C_N$ outputs a positive number $\epsilon$ and a list of primes $p_1(N),\ldots,p_{t_N}(N)$ such that there is an $\epsilon$-separated eigenbasis for $\{e^{iT(p_j(N))/\sqrt{p_j(N)}}\}_{j=1}^{t_N}$.
Let $g_{N,p}$ be the complexity of computing $e^{iT(p)/\sqrt{p}}$.
Then there is an algorithm that can solve Problem~\ref{prob:cusp-form-state} on input $N$ with complexity
$B_N +C_N + \sum_{j=1}^{t_N} g_{N,p_j(N)}$.
\end{prop}

\begin{proof}
  Given an instance $(N,f)$ of Problem~\ref{prob:cusp-form-state}, run algorithm $C$ with input $N$ to obtain $\epsilon > 0$ and primes $p_1(N),\ldots,p_{t_N}(N)$. Set $U_j= e^{iT(p_j(N))/\sqrt{p_j(N)}}$ and
$\alpha_j = e^{ia_{p_j(N)}/\sqrt{p_j(N)}}$, and run algorithm  $B$ with inputs
$N$, $\alpha_1,\ldots,\alpha_{t_N}$, $U_1, \ldots, U_{t_N}$ to obtain output $\ket{\psi}$. Suppose $\basis{0}$ is a solution to Problem~\ref{prob:cusp-form-state}; a solution exists since $V_N$ acted on by the $T(p)$ is isomorphic to $S_2(\Gamma_0(N))$ acted on by the $T_p$. For each $j$ the state $\basis{0}$ is an eigenvector for $U_j$ with eigenvalue $\alpha_j$. By Lemma~\ref{lem:unique-soln-cusp}(iii), $\ket{\psi}$ is a non-zero scalar multiple of $\basis{0}$, so it is a solution to Problem~\ref{prob:cusp-form-state}.

  For the complexity, algorithms $B$ and $C$ are each run once, and each $U_j$ is computed once.
\end{proof}
Theorem~\ref{GoldfeldHoffsteinThm} and Proposition~\ref{lemma:eps-prime-bis} guarantee that $\epsilon$ and
$p_1(N),\ldots,p_{t_N}(N)$ as in the above proof exist.
As in \S{}\ref{BrandtMatrixComputationSect}, each $U_j$ can be computed via a quantum algorithm with gate complexity that is polynomial in $p_j(N)$ and $\log(N)$.

\begin{prop}
An adversary that can solve Problem~\ref{prob:cusp-form-state} and has
a simultaneous eigenform $f$ for the Hecke operators can solve Problem~\ref{quat-Problem}.
\end{prop}

\begin{proof}
Solve Problem~\ref{prob:cusp-form-state} with input $f$ three times in succession.
\end{proof}

\section{Conclusion}

We have presented what seems 
to be a fairly efficient quantum money protocol. As far as we 
know, there are no subexponential attacks on this protocol, 
so it should be possible to implement securely with only a few hundred qubits.
We hope that the ideas and techniques of this paper could be used for other problems in cryptography and computer science. We also expect this paper to inspire work on the associated computational algebra problems.

\section{Acknowledgments} Kane would like to thank Scott Aaronson for his help with the presentation of \cite{Kane}. Sharif and Silverberg would like to thank John Voight for helpful discussions.
Kane was supported by NSF Award CCF-1553288 (CAREER), NSF Award CCF-2107547, a Sloan Research Fellowship, and a grant from CasperLabs. Silverberg was supported by NSF Award CNS-1703321 and a grant from the Alfred P.~Sloan Foundation.


\bibliographystyle{splncs04}

\appendix

\section{Lemmas}\label{AppA}

\begin{lem}
\label{lem:gener-double-eigenst}
  Suppose $V$ is a complex inner product space. Let $\{ e_1, \ldots, e_N\}$ and $\{f_1, \ldots, f_N\}$ be two real orthonormal bases for $V$. Then
    $\sum_{i=1}^Ne_i \otimes e_i = \sum_{i=1}^Nf_i \otimes f_i.$
\end{lem}

\begin{proof}
Let $\mathrm{End}(V)$ be the ring of linear maps $V \to V$. The map $\kappa: V \otimes V \to \mathrm{End}(V)$ given by $\kappa(v \otimes w)(x) = \braket{x, v}w$ is an isomorphism of real vector spaces. If $u$ is a real unit vector, then $\kappa(u \otimes u)$ is the projection operator onto $u$. It follows that both $\kappa(\sum e_i \otimes e_i)$ and $\kappa(\sum f_i \otimes f_i)$ are the identity map. The claim follows.
\end{proof}

The following two deterministic algorithms run in polynomial time.

\begin{algorithm}\label{lem:algor-liftinglem}\leavevmode
  INPUT: Positive integers $m$, $e$, and $r$ such that $e\mid m$ and $\gcd(r,e)=1$.

  OUTPUT: $k \in \Z$ such that $k \equiv r \pmod{e}$ and $\gcd(k,m)=1$.
  \begin{enumerate}[(1)]
      \item For $i=1,2,\ldots$ compute $d_i := \gcd(e^i,m)$ until $d_i = d_{i+1}$, and fix that $i$.
      \item Compute and output $k$ such that $k \equiv r \pmod{d_i}$ and $k \equiv 1 \pmod{m/d_i}$.
  \end{enumerate}
\end{algorithm}

Note that $\gcd(d_i,m/d_i)=1$. The above algorithm runs in polynomial time since $i \le \log_2 m$.

\begin{algorithm}\label{lem:algor-liftinglem2}\leavevmode
  INPUT: Positive integers $d$, $a$, and $b$ such that $\gcd(d,a,b)=1$.

  OUTPUT: $c \in \Z$ such that $c \equiv a \pmod{b}$ and $\gcd(d,c)=1$.
  \begin{enumerate}[(1)]
      \item Letting $e = \gcd(d,b)$, apply Algorithm~\ref{lem:algor-liftinglem} to compute  $c' \in \Z$ such that $c' \equiv a \pmod{e}$ and $\gcd(c',d)=1$.
      \item Apply the Chinese Remainder Theorem to compute $c \in \Z$ such that $c \equiv c' \pmod{d}$ and $c \equiv a \pmod{b}$.
  \end{enumerate}
\end{algorithm}

  \begin{prop}\label{prop:algor-generatorm}
There is a deterministic polynomial-time algorithm that, giv\-en a positive integer $m$ and a cyclic subgroup $H \subset (\Z/m\Z)^2$ of order $m$, computes $(d,c) \in \Z^2$ that generates $H$ and satisfies $d \mid m$ and $\gcd(d,c)=1$.
The integer $d$ is the unique divisor of $m$ such that $(d,\gamma)$ generates $H$ for some $\gamma\in\Z$.
  \end{prop}

  \begin{proof}
Theorem 2.6.9 of \cite{Teodorescu} gives a deterministic polynomial-time algorithm that given $m$ and $H$, produces $(d',c') \in \Z^2$ that generates $H$. Since $H$ has order $m$ we have $\gcd(d',c',m)=1$.
Let $d = \gcd(d',m)$. Compute integers $r,s$ such that $d=rd'+sm$.
Letting $e = m/d$, then $\gcd(r,e)=1$.
Apply Algorithm~\ref{lem:algor-liftinglem} in Appendix~\ref{AppA} to compute
$k \in\Z$ such that $k \equiv r \pmod{e}$ and $\gcd(k,m)=1$.
Then $kd' \equiv rd' \pmod{md'/d}$, so $kd' \equiv rd' \equiv d \pmod{m}$.
Let  $c = kc' \pmod{m}$. Then $(d,c) = k(d',c')$ generates $H$.
Since $d \mid m$ we have $\gcd(d,c)=\gcd(d,c,m)=1$.
Since $i \le \log_2 m$, the algorithm runs in polynomial time.

Projecting $H$ onto the first component gives a cyclic subgroup of $\Z/m\Z$ of order $m/d$, for which $d$ is the unique generator that divides $m$.
  \end{proof}

The next result gives an algorithm to compute an isomorphism
$\co_N/m\co_N \isom M_2(\Z/m\Z)$, where $\co_N$ is a maximal order and the prime $N \nmid m$.
For our purposes, it is important that the algorithm produce the same isomorphism each time it is given the same inputs $N$, $\co_N$, and $m$.
The algorithm invokes a polynomial-time quantum algorithm to factor $m$. As such, there is
some small failure probability.
After that, it uses a classical polynomial-time algorithm due to Voight to
deterministically construct isomorphisms $\co_N/p^r\co_N \isom M_2(\Z/p^r\Z)$
for each prime divisor $p$ of $m$, where $p^r || m$.

  \begin{prop}\label{prop:algor-M2misomalgor}
 There is an algorithm in complexity class BQP that, given a positive integer $m$,  a prime $N$ that does not divide $m$, a maximal order $\co_N$ in $H_N$,
  and a $\Z$-basis for $\co_N$, produces an isomorphism $$f_{N,m}: \co_N/m\co_N \isom M_2(\Z/m\Z).$$
  \end{prop}

  \begin{proof}
Factor $m$ (for example using
Shor's algorithm).
For each prime divisor $p$ of $m$, with $p^r || m$,
apply Proposition 4.8 of \cite{VoightArticle} and the results mentioned in the paragraph
after Problem 4.9 of \cite{VoightArticle}
to obtain an isomorphism $\co_N/p^r\co_N \isom M_2(\Z/p^r\Z)$.
Then apply the Chinese Remainder Theorem.
  \end{proof}

  \begin{lem}\label{lem:rowspacelem}
Fix $m\in\Z^{>0}$.
Let $f$ denote the map from the set of left ideals of $M_2(\Z/m\Z)$
to the set of subgroups of $(\Z/m\Z)^2$ induced by sending a matrix to its rowspace.
Then $f$ is a bijection, and its inverse is the map $g$ that sends a subgroup $H$ to the set of matrices whose rows are in $H$.
  \end{lem}

\begin{proof}
Suppose that $H$ is a subgroup of $(\Z/m\Z)^2$.
The left action of $M_2(\Z/m\Z)$ on $H$ is by row operations, so if $A \in M_2(\Z/m\Z)$
and $B \in g(H)$, then the rows of $AB$ are linear combinations of the rows
of $B$, so $AB \in g(H)$.
Thus $g(H)$ is a left ideal, and $fg$ is the identity.

To show that $gf$ is the identity, suppose $\ib$ is a left ideal of $M_2(\Z/m\Z)$ and
let $H = f(\ib)$.
Then $\ib \subset g(H)$. To show $g(H) \subset \ib$,
suppose $(x,y) \in H$.
By the definition of $H$, there are matrices $A_1,...,A_r \in \ib$ and
for each $i$ a row $a_i$
 of $A_i$ such that $(x,y) = \sum_{i=1}^r a_i$. Left-multiplying $A_i$ by
$\left[\begin{smallmatrix} 0 & 1 \\ 1 & 0 \end{smallmatrix}\right]$ if
necessary, we may assume that $a_i$ is the top row of $A_i$.
Then
$\left[\begin{smallmatrix} x & y \\ 0 & 0 \end{smallmatrix}\right] =
\left[\begin{smallmatrix} 1 & 0 \\ 0 & 0 \end{smallmatrix}\right] \sum_{i=1}^r A_i  \in \ib$, and
$\left[\begin{smallmatrix} 0 & 0 \\ x & y \end{smallmatrix}\right] =
\left[\begin{smallmatrix} 0 & 1 \\ 1 & 0 \end{smallmatrix}\right]
\left[\begin{smallmatrix} x & y \\ 0 & 0 \end{smallmatrix}\right]  \in \ib$.
Since such matrices generate $g(H)$, we have  $g(H) \subset \ib$.
\end{proof}

  \begin{lem}\label{lem:HJlem}
Suppose that $d,b\in\Z^{>0}$, that $c,c'\in\Z$, that $c = c' \pmod{b}$, and that $\gcd(d,c)=1=\gcd(d,c')$.
Let $m=db$ and suppose that $H$ and $H'$ are the subgroups of $(\Z/m\Z)^2$ generated by $(d,c)$ and by $(d,c')$, respectively.
Then $H = H'$.
  \end{lem}

\begin{proof}
Since $\gcd(d,c)=1$,
there exist integers $x$ and $y$ such that $cx = 1 + d y$. Setting
$\lambda  = 1 + x(c'-c)$,
then $\lambda c =
c' + y d(c'-c)$, and
since $c' \equiv c \pmod{b}$
it follows that $\lambda  c \equiv c' \pmod{m}$
and $\lambda d \equiv d \pmod{m}$.
Thus $\lambda (d,c) = (d,c')$ in $(\Z/m\Z)^2$, so
$H' \subset H$. By symmetry, $H \subset H'$, so $H = H'$.
\end{proof}

Recall (Definition \ref{RightOrderWtDef}) that if $I$ is a left fractional $\co_N$-ideal of the quaternion algebra $H_N$, then we let $\co_I$ be its right order and
we let $w_{I} = \#(\co_I^\times/\{\pm 1\})=\frac{1}{2}\#\co_I^\times$.
An integral solution to $x^2 - 3y^2 = -N$ can be found in polynomial time by~\cite{Simon}.
\begin{prop}\label{prop:constants-w_i}
  We have $w_{I} = 1$ for all $[I] \in \cl(\co_N)$, with the following exceptions:
  \begin{enumerate}[(i)]
      \item If $N \equiv 5 \pmod{12}$, then $w_{I} = 3$ for all $I\in [\co_N]$.
      \item If $N \equiv 7 \pmod{12}$, then $w_{I} = 2$ for all $I\in [\co_N]$.
      \item Suppose $N \equiv 11 \pmod{12}$. Let $(a,b) \in \Z^2$ be a solution to $x^2 - 3y^2 = -N$.
Let $\alpha := \frac{a}{3b}i + \frac{1}{3b}ij$ and $\hat{\co} := \Z + \frac{1+j}{2}\Z + \alpha \Z + \frac{\alpha - \alpha j}{2} \Z$.
Then
      $w_{I} = 3$ for all $I\in [\co_N]$, and
      $w_{I} = 2$ for all $I\in [\co_N \cdot \hat{\co}]$.
  \end{enumerate}
\end{prop}

\begin{proof}
  Table~1.3 of \cite{Gross} shows that $w_{I} = 1$ for all $[I]$, with the following exceptions: when $N \equiv 5 \pmod{12}$ one ideal class satisfies $w_{I} = 3$; when $N \equiv 7 \pmod{12}$ one ideal class satisfies $w_{I} = 2$; and when $N \equiv 11 \pmod{12}$ there are two ideal classes $[I]$ and $[J]$ such that $w_{I} = 3$ and $w_{J} = 2$.
  If $N \equiv 7 \pmod{12}$ then $i^2 = -1$ so $i \in \co_N^\times$ has order $4$, and hence $w_{\co_N} = 2$.
  If $N \equiv 5 \pmod{6}$ then $i^2 = -3$ so $\frac{1+i}{2} \in \co_N^\times$ has order $6$, and hence $w_{\co_N} = 3$.
  Now suppose $N \equiv 11 \pmod{12}$.
One can check that $\hat{\co}$ is a maximal order, and $\hat{\co}$ is the right order of the ideal $\co_N \cdot \hat{\co}$.
  We have
$\alpha^2
       = \frac{-(a^2 + N)}{3b^2} = -1$,
so $\alpha \in \hat{\co}^\times$ has order $4$ and hence $w_{\co_N \cdot \hat{\co}} = 2$.
\end{proof}
See also \cite{ibukiyama}, where $\co_N$ is denoted $\co(3)$ when $N \equiv 5 \pmod{6}$, and $\hat{\co}$ is denoted $\co'(1)$.

\section{$\epsilon$-separation Data}
\label{AppB}

  \begin{table}[ht]
\caption{$\epsilon$-separation for $\{e^{iT(p)/\sqrt{p}} \mid p < \log_2(N)\}$}
\label{tab:epsilon-separation-n}
\begin{center}
\begin{tabular}[t]{rl}
\toprule
$N$ & $\epsilon$\\
\midrule
547 & 0.4824236848637427 \\
557 & 0.7199773703667618 \\
563 & 0.7553525215246627 \\
569 & 0.9200347021863563 \\
571 & 0.48205861423463164\\
577 & 0.40674046098264244\\
587 & 0.7982583121867862 \\
593 & 0.9266761931828437 \\
599 & 0.62563971482572   \\
601 & 0.7182238262429224 \\
607 & 0.7313809878961292 \\
613 & 0.768492003890778  \\
617 & 0.5983414655675874 \\
619 & 0.6187541297546084 \\
631 & 0.45419000886679206\\
641 & 0.43490142944562354\\
643 & 0.6346083766649872 \\
647 & 0.7432521901131    \\
653 & 0.5063114409620633 \\
659 & 0.6777125171096566 \\
\bottomrule
\end{tabular}
\quad \quad \quad
\begin{tabular}[t]{rl}
\toprule
$N$ & $\epsilon$\\
\midrule
12569 & 0.22159756788222007\\
12577 & 0.22690747823008486\\
12583 & 0.2774346724081338 \\
12589 & 0.22865081262562248\\
12601 & 0.25482871813162855\\
12611 & 0.16451483770778993\\
12613 & 0.09017383560136713\\
12619 & 0.18211198468203824\\
12637 & 0.16246553818517484\\
12641 & 0.19366213429958556\\
\midrule
20011 & 0.34309639146812015\\
20021 & 0.3536950173591149 \\
20023 & 0.2610129987276544 \\
20029 & 0.19283243271645334\\
20047 & 0.30798681044672843\\
20051 & 0.2711650765294632 \\
20063 & 0.21456144876447153\\
20071 & 0.3506564319413416 \\
20089 & 0.2942067355453101 \\
\bottomrule
\end{tabular}

\end{center}
\end{table}

For each prime $N$ in Table~\ref{tab:epsilon-separation-n}, let $p_1,\ldots,p_t$ be the primes less than $\log_2(N)$, and set $U_j = e^{iT(p_j)/\sqrt{p_j}}$. 
Letting $\basis{1},\ldots,\basis{h} \in V_N$ be the simultaneous eigenvectors for the $U_j$'s, we used Sage to compute the corresponding tuples of eigenvalues $v_1,\ldots,v_h$, and the minimum Euclidean distance between pairs of tuples of eigenvalues.
In Table~\ref{tab:epsilon-separation-n}, the value $\epsilon$ is the minimum Euclidean distance $|v_i - v_j|$ for $i \neq j$, and therefore is the largest value of $\epsilon$ for which the eigenbasis is $\eps$-separated. 
The Sage code we used to generate the table is publicly available at \url{https://github.com/ssharif/QuantumMoneyCode}.


\end{document}